\documentclass[runningheads]{llncs}
\makeatletter  %% this is crucial 
 \renewcommand\subsubsection{
  \@startsection{subsubsection}{2}{\z@}% 
  {-18\p@ \@plus -4\p@ \@minus -4\p@}% 
  {8\p@ \@plus 4\p@ \@minus 4\p@}%     
  {\normalfont\normalsize\bfseries\boldmath \rightskip=\z@ \@plus 8em \pretolerance=10000 }
} 
\makeatother   %% this is crucial 

 \usepackage{fancyvrb}
 \usepackage{times}
 \usepackage{srcltx}
 \usepackage{braket}
 \usepackage{tikz}
 \usepackage{verbatim}
 \usepackage{hyperref}
 \usepackage{longtable}
 \usepackage{color}
 \usepackage{stmaryrd} % for \bigsqcap

 \usepackage{amsmath}
 \usepackage{amssymb}
 \usepackage{latexsym}
 
\usepackage[british]{babel}
\usepackage[latin1]{inputenc}
\usepackage{xcolor}

\newcommand{\PS}{\texttt{PolicySet}}
\newcommand{\Pol}{\texttt{Policy}}
\newcommand{\R}{\texttt{Rule}}
\newcommand{\Tar}{\texttt{Target}}
\newcommand{\Cnd}{\texttt{Condition}}
\newcommand{\Rq}{\texttt{Request}}
\newcommand{\Mtc}{\texttt{Match}}
\newcommand{\Any}{\texttt{AnyOf}}
\newcommand{\All}{\texttt{AllOf}}

\newcommand{\Deny}{\texttt{Deny}}
\newcommand{\Permit}{\texttt{Permit}}

\newcommand{\is}{::=}

\newcommand{\effect}{\textit{Effect}}

\newcommand{\deny}{\mathbf{d}}
\newcommand{\permit}{\mathbf{p}}

\newcommand{\T}{\top}
\newcommand{\Tp}{\top_{\permit}}
\newcommand{\Td}{\top_{\deny}}

\newcommand{\F}{\bot}
\newcommand{\I}{I}

\newcommand{\Id}{I_{\deny}}
\newcommand{\Ip}{I_{\permit}}
\newcommand{\Idp}{I_{\deny\permit}}

\newcommand{\true}{\mathit{true }}
\newcommand{\false}{\mathit{false }}

\newcommand{\policyvalues}{\mathbf{P}}

\newcommand{\four}{\mathbf{four }}
\newcommand{\tval}{\mathbf{t\!t}} 
\newcommand{\fval}{\mathbf{f\!f}}
\newcommand{\topf}{\top\!\top}
\newcommand{\botf}{\bot\!\bot}

\newcommand{\pval}{\mathbf{p}} 
\newcommand{\dval}{\mathbf{d}}
\newcommand{\naa}{\mathbf{\frac{n}{a}}} %for d-algebra

\newcommand{\semantics}[1]{ \llbracket #1\rrbracket}
\newcommand{\p}{p}
\renewcommand{\d}{d}
\renewcommand{\max}{\mathit{max}}
\renewcommand{\min}{\mathit{min}}
\newcommand{\mymax}{\mathit{Max}_{\mysubseteq}}
\newcommand{\mymin}{\mathit{Min}_{\mysubseteq}}
\newcommand{\eval}{\mathit{eval}}
\newcommand{\comb}{\theta}
\newcommand{\mysubseteq}{\sqsubseteq_{\policyvalues}}
\newcommand{\posetP}{\boldsymbol{P}_{\policyvalues}}

\newcommand{\po}{\mathbf{p-o}} %permit overrides
\newcommand{\denyo}{\mathbf{d-o}} %deny overrides
\newcommand{\fa}{\mathbf{f-a}} %first applicable
\newcommand{\oa}{\mathbf{o-1-a}}%only one applicable 

\newcommand{\oplusb}{\oplus^{B}}
\newcommand{\otimesb}{\otimes^{B}}
\newcommand{\oplusd}{\oplus^{\mathcal{D}}}
\newcommand{\otimesd}{\otimes^{\mathcal{D}}}
\newcommand{\odotd}{\odot^{\mathcal{D}}}
\newcommand{\ominusd}{\ominus^{\mathcal{D}}}

\newcommand{\mc}{\mathcal}
\newcommand{\mb}{\mathbb}
\newcommand{\mf}{\mathfrak}

\newcommand{\calL}{\mathcal{L}}
\newcommand{\calP}{\mathcal{P}}

\newcommand{\D}{$\mathcal{D}$}
\newcommand{\half}{\frac{1}{2}}

\newcommand{\funnyarrow}{\leadsto}
\newcommand{\lubL}{\bigsqcup}
\newcommand{\glbL}{\bigsqcap}

\newcommand{\mynote}[1]{\textcolor{red}{\textbf{#1}}}
\newcommand{\etal}{\textit{et al. }}
\newcommand{\myangle}[1]{\langle #1 \rangle}
\newcommand{\seq}[1]{\langle #1 \rangle}
\newcommand{\ra}{\rightarrow}

\usepackage{url}
\urldef{\mailsa}\path|{cdpu, riis, nielson}@imm.dtu.dk|

\begin{document}

%\mainmatter  % start of an individual contribution

\title{The Logic of XACML -- Extended}

\titlerunning{The Logic of XACML}

\author{Carroline Dewi Puspa Kencana Ramli, Hanne Riis Nielson, Flemming Nielson}
\authorrunning{Carroline Dewi Puspa Kencana Ramli, Hanne Riis Nielson, Flemming Nielson}

\institute{Department of Informatics and Mathematical Modelling \\
Danmarks Tekniske Universitet \\
Lyngby, Denmark\\
\email{\mailsa}}

\maketitle

\begin{abstract}
We study the international standard XACML 3.0 for describing security access control policy in a compositional way. 
Our main contribution is to derive a logic that precisely captures the idea behind the standard and to formally define the semantics of the policy combining algorithms of XACML.
To guard against modelling artefacts we provide an alternative way of characterizing the policy combining algorithms and we formally prove the equivalence of these approaches.
This allows us to pinpoint the shortcoming of previous approaches to formalization based either on Belnap logic or on \D-algebra. 
\end{abstract}

% !TEX root = ../facs-extended.tex
\section{Introduction}
\label{s:introduction}
XACML (eXtensible Access Control Markup Language) is an approved OASIS \footnote{OASIS (Organization for the Advancement of Structured Information Standard) is a non-for-profit, global consortium that drives the development, convergence, and adoption of e-business standards. Information about OASIS can be found at \url{http://www.oasis-open.org}.} Standard access control language \cite{xacml,XACMLSpesification}. XACML describes both an access control policy language and a request/response language. The \emph{policy language} is used to express access control policies (\textit{who can do what when}) while the \emph{request language} expresses queries about whether a particular access should be allowed  and the \emph{response language} describes answers to those queries.

In order to manage modularity in access control, XACML constructs policies into several components, namely \emph{PolicySet}, \emph{Policy} and \emph{Rule}. A PolicySet is a collection of other PolicySets or Policies whereas a Policy consists of one or more Rules. A Rule is the smallest component of XACML policy and each Rule only either grants or denies an access. As an illustration, suppose we have access control policies used within a National Health Care System. The system is composed of several access control policies of local hospitals. Each local hospital has its own policies such as patient policy, doctor policy, administration policy, etc. Each policy contains one or more particular rules, for example, in patient policy there is a rule that only the designated patient  can read his or her record. In this illustration, both the National Health Care System and local hospital policies are  PolicySets. However the patient policy is a Policy and one of its rules is the patient record policy. Every policy is only applicable to a certain target and a policy is applicable when a request matches to its target, otherwise, it is not applicable. The evaluation of composing policies is based on a combining algorithm -- the procedure for combining decisions from multiple policies. There are four standard combining algorithms in XACML i.e., (i) permit-overrides, (ii) deny-overrides, (iii) first-applicable and (iv) only-one-applicable.

The syntax of XACML is based on XML format \cite{xml}, while its standard semantics is described normatively using natural language in \cite{XACMLSpesification}. Using English paragraphs in standardization leads to misinterpretation and ambiguity. In order to avoid this drawback, we define an abstract syntax of XACML 3.0 and a formal XACML components evaluation based on XACML 3.0 specification in Section \ref{s:XACML components}. Furthermore, the evaluation of the XACML combining algorithms is explained in Section \ref{s:combining algorithms}.

Recently there are some approaches to formalizing the semantics of XACML. In \cite{Halpern2008}, Halpern and Weissman show XACML formalization using First Order Logic (FOL). However, their formalization does not capture whole XACML specification. It is too expensive to express XACML combining algorithms in FOL. Kolovski \etal in \cite{Kolovski2007,Kolovski2007a} maps a large fragment of XACML to Description Logic (DL) -- a subset of FOL -- but they leave out the formalization of only-one-applicable combining algorithm. Another approach is to represent XACML policies  in term of Answer Set Programming (ASP). Although Ahn \etal in  \cite{Ahn2010} show a complete XACML formalization in ASP, their formalization is based on XACML 2.0, which is out-of-date nowadays. More particular, the combining algorithms evaluation in XACML 2.0 is simpler than XACML 3.0. Our XACML 3.0 formalization is closer to multi-valued logic approach such as Belnap logic \cite{Belnap1977} and \D-algebra \cite{Ni2009}. Bruns \etal in \cite{Bruns2007,Bruns2008}  and Ni \etal in \cite{Ni2009} define a logic for XACML using Belnap logic and \D-algebra, respectively. In some cases, both methods show different results from the XACML standard specification. We discuss the shortcoming of formalization based either on Belnap logic or on \D-algebra  in Section \ref{s:related work} and we conclude in Section \ref{s:conclusion and future work}.

% !TEX root = ../facs-extended.tex

%=================================XACML Components=================================
\section{XACML Components}
\label{s:XACML components}
XACML syntax is described verbosely in XML format. For our analysis purpose, we do abstracting XACML components.  From the abstraction, we show how XACML evaluates policies. We give an example how XACML policies can be described in our abstraction and the components evaluation at the end of this section.  

\subsection{Abstracting XACML Components}
There are three main policy components in XACML, namely \PS, \Pol\ and \R. \PS\ is the root of all XACML policies. A \PS\ is composed of a sequence of other \PS\ or \Pol\ components along with a policy combining algorithm ID and a \Tar. A \Pol\ is composed of a sequence of \R, a \Tar\ and a rule combining algorithm ID. A \R\ is a single entity that defines the individual rule in the policy. Each \R\  has a particular effect to an access request, i.e., either \emph{deny} or \emph{permit} the access.  Each \R\  is composed of a \Tar\ and a \Cnd.  A \Tar\ is an XACML component that indicates under which  categories an XACML policy is applicable. A \Tar\ consists of  conjunction of \Any\  component  with each \Any\ consists of disjunction of \All\ components and each \All\  consists of conjunction of \Mtc. Each \Mtc\ contains only one particular category to be matched with the request.  Typical categories of XACML attributes  are \textit{subject} category (e.g. human user, workstation, etc) \textit{action} category (e.g. read, write, delete, etc),  \textit{resource} category (e.g. database, server, etc) and \textit{environment} category (e.g. SAML, J2SE, CORBA, etc). A \Cnd\ is a set of propositional formulae that refines the applicability of a \R. 

A \Rq\ contains a set of available informations on desired access request such as subject, action, resource and environment categories. A \Rq\ also contains additional information about external state, e.g. the current time, the temperature, etc.

We present in Table \ref{t:syntax} a succinct syntax of XACML 3.0 that is faithful to the more verbose syntax used in the standard \cite{XACMLSpesification}.

\vspace{-15pt} 
\begin{table}[h]
\caption{Abstraction of XACML 3.0 Components}
\label{t:syntax}
\vspace{-10pt} 
\begin{center}

\begin{tabular}{|lcll|}
\hline
  \multicolumn{4}{|c|}{\textbf{XACML Policy Components}} \\
   \PS      & \is & $\myangle{\textrm{\Tar},\myangle{\textrm{\PS}_1, \ldots, \textrm{\PS}_m},\comb}$ & \\
             & \textbar & $\myangle{\textrm{\Tar},\myangle{\textrm{\Pol}_1, \ldots, \textrm{\Pol}_m},\comb}$ & where $m \geq 0$ \\
   \Pol       & \is & $\myangle{ \textrm{\Tar},\myangle{\textrm{\R}_1, \ldots, \textrm{\R}_m},\comb}$ & where $m \geq 1$ \\
   \R       & \is & $\myangle{\effect, \textrm{\Tar}, \textrm{\Cnd}}$ & \\ 
   \Cnd   & \is & \textit{propositional formulae}  & \\
   \Tar & \is & $\mathbf{Null}$ & \\
             & \textbar & $\textrm{\Any}_1 \wedge \ldots \wedge \textrm{\Any}_m$ & where $m \geq 1$ \\
	\Any  & \is & $\textrm{\All}_1 \vee \ldots \wedge \textrm{\All}_m$ & where $m \geq 1$ \\
   \All  & \is & $\textrm{\Mtc}_1 \wedge \ldots \wedge \textrm{\Mtc}_m$ & where $m \geq 1$ \\
   \Mtc  & \is & $\Phi(\alpha)$ & \\ %match
	$\Phi  $  & \is & $\mathbf{subject}$ \textbar\ $\mathbf{action}$ \textbar\ $\mathbf{resource}$ \textbar\ $\mathbf{enviroment}$ &\\
   $\alpha$  & \is & \textit{attribute value} & \\
   $\theta$  & \is & $\po$ \textbar\ $\denyo$ \textbar\ $\fa$ \textbar\ $\oa$ & \\
   $\effect$ & \is & $\deny$\ \textbar\ $\permit$  &\\ 
\hline
	\multicolumn{4}{|c|}{\textbf{XACML Request Component}} \\
   %$\Req$    & \is & $\Set{\Phi_1(\alpha_1), \ldots, \Phi_m(\alpha_m)}$ & where $m \geq 1$ \\
  \Rq   & \is & $\Set{A_1, \ldots, A_m}$ & where $m \geq 1$ \\
  $A$			& \is & $\Phi(\alpha)$ \textbar\ \textit{external state} &\\
\hline
   \end{tabular}
\end{center}
\vspace{-15pt} 
\end{table}

\subsection{XACML Evaluation}
The evaluation of XACML components starts from \Mtc\ evaluation  and it is continued  iteratively until  \PS\ evaluation. The \Mtc, \All, \Any, and \Tar\ values are either \textit{match}, \textit{not match} or \textit{indeterminate}. The value is indeterminate if there is an error during the evaluation so that the decision cannot be made at that moment.  The \R\ evaluation depends on \Tar\ evaluation and \Cnd\ evaluation. The \Cnd\ component is a set of propositional formulae which each formula is evaluated to either \textit{true}, \textit{false} or \textit{indeterminate}. An empty \Cnd\ is always evaluated to \textit{true}. The \R's value is either  \textit{applicable}, \textit{not applicable} or \textit{indeterminate}. An applicable \R\ has effect either \textit{deny} or \textit{permit}. Finally, the evaluation of \Pol\ and \PS\ are based on a combining algorithm of which the result can be either  \textit{applicable} (with its effect either \textit{deny} or \textit{permit}), \textit{not applicable} or \textit{indeterminate}.

\subsubsection{Three-Valued Lattice}

We use three-valued logic to determine XACML evaluation value. 
We define $\calL_3 = \myangle{V_3, \leq}$ be \textit{three-valued lattice} where $V_3$ is the set $\Set{\T, \I, \F}$ and  $\F \leq \I \leq \T$. %$\leq$ is as in Figure \ref{f:L3}. 
Given a subset $S$ of $V_3$, we denote the greatest lower bound (glb) and the least upper bound (lub) at $S$ (w.r.t. $\calL_3$) by $\glbL S$ and $\lubL S$, respectively. Recall that $\glbL \emptyset = \top$ and $\lubL \emptyset = \bot$.

We use $\semantics{.}$ notation to map XACML elements into their evaluation values. The evaluation of XACML components to values in $V_3$ is summarized in Table \ref{t:mapping v3 xacml}.
\begin{table}
\vspace{-15pt} 
\caption{Mapping $V_3$ into XACML Evaluation Values}
\vspace{-15pt} 
\label{t:mapping v3 xacml}
\[
\begin{array}{|c|c|c|c|}
\hline
V_3 & \textrm{\Mtc\ and \Tar\ value} & \textrm{\Cnd\ value} & \textrm{\R, \Pol\ and \PS value} \\
\hline
\top & \textrm{match} & \textrm{true} & \textrm{applicable (either deny or permit)} \\
\bot & \textrm{not match} & \textrm{false} & \textrm{not applicable} \\
\I   & \textrm{indeterminate} & \textrm{indeterminate} & \textrm{indeterminate}\\
\hline
\end{array}
\]
\vspace{-15pt} 
\end{table}

\subsubsection{Match Evaluation}
A \Mtc\ element $\mc{M}$ is an attribute value that the request should fulfill. Given a \Rq\ component $\mc{Q}$, the evaluation of \Mtc\ element is as follows:
\begin{equation}
   \semantics{\mc{M}}(\mc{Q}) = 
   \begin{cases}
   \top & \mc{M} \in \mc{Q} \\
   \bot & \mc{M} \not\in \mc{Q} \\
   \I   & \textrm{there is an error during the evaluation} \\
   \end{cases}
\end{equation}

\subsubsection{Target Evaluation}
Let $\mc{M}$ be a \Mtc, $\mc{A} = \mc{M}_1 \wedge \ldots \wedge \mc{M}_m$ be an \All, $\mc{E} = \mc{A}_1 \vee \ldots \vee \mc{A}_n$ be an \Any,  $\mc{T} = \mc{E}_1 \wedge \ldots \mc{E}_o$ be a \Tar\  and $\mc{Q}$ be a \Rq. Then, the evaluations of \All, \Any, and \Tar\ are as follows:

\begin{equation}
\semantics{\mc{A}}(\mc{Q}) = \bigsqcap_{i=1}^m \semantics{\mc{M}_i}(\mc{Q})
\end{equation}

\begin{equation}
\semantics{\mc{E}}(\mc{Q}) = \bigsqcup_{i=1}^n \semantics{\mc{A}_i}(\mc{Q})
\end{equation}

\begin{equation}
\semantics{\mc{T}}(\mc{Q}) = \bigsqcap_{i=1}^o \semantics{\mc{E}_i}(\mc{Q})
\end{equation}

In summary, we can simplify the \Tar\ evaluation as follows:
\begin{equation}
 \semantics{\mc{T}}(\mc{Q}) = \bigsqcap \bigsqcup \bigsqcap \semantics{\mc{M}}(\mc{Q})
\end{equation}

An empty \Tar\ -- indicated by $\mathbf{Null}$ -- is always evaluated to $\top$. 

\subsubsection{Condition Evaluation}
We define the conditional evaluation function $\eval$ as an arbitrary function to evaluate \Cnd\ to value in $V_3$ given  a \Rq\ component $\mc{Q}$. The evaluation of \Cnd\ is defined as follows:

\begin{equation}
\semantics{\mc{C}}(\mc{Q}) = \eval(\mc{C},\mc{Q})
\end{equation}

\subsubsection{Extended Values}
In order to distinguish an applicable policy to permit an access from applicable policy to deny an access, we extend $\top$ in $V_3$ value to $\Tp$ and $\Td$, respectively. The same case also applies to indeterminate value. The extended indeterminate value contains the potential effect values which could have occurred if there would not have been an error during a evaluation. The possible extended indeterminate values are \cite{XACMLSpesification}:
\begin{itemize}
 \item Indeterminate Deny ($\Id$): an indeterminate from a policy  which could have evaluated to deny but not permit, e.g., a \R\ which evaluates to indeterminate and its effect is deny.
 \item Indeterminate Permit ($\Ip$): an indeterminate from a policy which could have evaluated to permit but not deny, e.g., a \R\ which evaluates to indeterminate and its effect is permit.
 \item Indeterminate Deny Permit ($\Idp$): an indeterminate from a policy which could have effect either deny or permit.
\end{itemize}
We extend the set $V_3$ to $V_6  = \Set{\Tp, \Td, \Id, \Ip, \Idp, \bot}$ and we use $V_6$  to evaluate XACML policies. 

\subsubsection{Rule Evaluation}
Let $\mc{R} = \myangle{*,\mc{T},\mc{C}}$ be a \R\ and $\mc{Q}$ be a \Rq. Then, the evaluation of \R\ is determined as follows:
\begin{equation}
\label{e:rule evaluation}
   \semantics{\mc{R}}(\mc{Q}) = 
\begin{cases}
   \top_{*}   & \semantics{\mc{T}}(\mc{Q}) = \top \textrm{ and } \semantics{\mc{C}}(\mc{Q}) = \T \\
   \F & \bigl( \semantics{\mc{T}}(\mc{Q}) = \T \textrm{ and } \semantics{\mc{C}}(\mc{Q}) = \F \bigr) \textrm{ or }  \semantics{\mc{T}}(\mc{Q}) = \F \\
   \I_{*} & \textrm{otherwise} \\
\end{cases}
\end{equation}

Let $F$ and $G$ be two values in $V_3$. We define a new operator $\funnyarrow : V_3 \times V_3 \ra V_3$ as follows:
\begin{equation}
\label{e:funnyarrow}
   F \funnyarrow G = 
\begin{cases}
   G &  \textrm{if } F = \T \\
   F & \textrm{otherwise}
\end{cases}
\end{equation}

We define a function $\sigma:V_3 \times \Set{\permit, \deny} \ra V_6$ that maps a value in $V_3$ into a value in $V_6$ given a particular \R's effect as follows:
\begin{equation}
\sigma(X,*) =
\begin{cases}
\label{e:sigma}
 X   & \textrm{if } X = \bot \\
 X_* & \textrm{otherwise}
\end{cases}
\end{equation}

\begin{proposition}
\label{prop: rule evaluation}
Let $\mc{R} = \myangle{*,\mc{T},\mc{C}}$ be a \R\ and  $\mc{Q}$ be a \Rq. Then, the following equation holds
\begin{equation}
   \semantics{\mc{R}}(\mc{Q}) = \sigma\left(\semantics{\mc{T}}(\mc{Q}) \funnyarrow \semantics{\mc{C}}(\mc{Q}), * \right)
\end{equation} 
\end{proposition}
\begin{proof}
The table below shows the proof of Proposition  \ref{prop: rule evaluation}.
\[\begin{array}{c|c|c|c|c}
\semantics{\mc{T}}(\mc{Q}) & \semantics{\mc{C}}(\mc{Q}) & \semantics{\mc{T}}(\mc{Q}) \funnyarrow \semantics{\mc{C}}(\mc{Q}) &  \sigma\left(\semantics{\mc{T}}(\mc{Q}) \funnyarrow \semantics{\mc{C}}(\mc{Q}), * \right) & \semantics{\mc{R}}(\mc{Q}) \\  
\hline
\top & \top & \top & \top_* & \top_* \\
\top & \bot & \bot & \bot   & \bot \\
\top & \I   & \I   & \I_*   & \I_* \\
\hline
\bot & \top & \bot & \bot & \bot \\
\bot & \bot & \bot & \bot & \bot \\
\bot & \I   & \bot & \I_*  & \I_* \\
\hline
\I & \top & \I & \I_* & \I_* \\
\I & \bot & \I & \I_* & \I_* \\
\I & \I & \I & \I_* & \I_* \\
  \end{array} 
\]
\hfill$\Box$
\end{proof}

\subsubsection{Policy Evaluation}
The standard evaluation of Policy element taken from \cite{XACMLSpesification} is as follows:
\begin{center}
\small{
\begin{tabular}{|c|c|c|}
\hline
\Tar\ value & \R\ value & \Pol\ Value\\
\hline
match         & At least one \R\ value is applicable & Specified by the combining algorithm \\
match         & All \R\ values are not applicable    & not applicable \\
match         & At least one \R\ value is indeterminate        & Specified by the combining algorithm \\
not match     & Don't care & not applicable \\
indeterminate & Don't care & indeterminate  \\
\hline
\end{tabular}
}
\end{center}

Let $\mc{P} = \myangle{\mc{T},\mb{R},\comb}$ be a \Pol\  where $\mb{R} = \seq{\mc{R}_1, \ldots, \mc{R}_n}$. Let $\mc{Q}$ be a \Rq\  and $\mb{R'} = \seq{\semantics{\mc{R}_1}(\mc{Q}), \ldots, \semantics{\mc{R}_n}(\mc{Q})}$. The evaluation of \Pol\ is defined as follows:
\begin{equation}
\label{e:ext policy semantics}
   \semantics{\mc{P}}(\mc{Q}) = 
\begin{cases}
   \I_*      & \semantics{\mc{T}}(\mc{Q}) = \I \textrm{ and } \bigoplus_\comb(\mb{R'}) \in  \Set{\T_*,\ \I_*} \\
   \bot      & \semantics{\mc{T}}(\mc{Q}) = \F \textrm{ or}\\
             & \semantics{\mc{T}}(\mc{Q}) = \T \textrm{ and }\forall \mc{R}_i: \semantics{\mc{R}_i}(\mc{Q}) = \F \\
\bigoplus_\comb(\mb{R'}) & \textrm{otherwise} 
\end{cases}   
\end{equation}
\note{The combining algorithms denoted by $\bigoplus$ is explained in Section \ref{s:combining algorithms}.}

\subsubsection{PolicySet Evaluation}
The evaluation of \PS\ is similar to \Pol\ evaluation. However, the input of the combining algorithm is a sequence of either \PS\ or \Pol\ components. 

Let $\mc{PS} = \myangle{\mc{T},\mb{P},\comb}$ be a \PS\ where $\mb{P} = \seq{\mc{P}_1, \ldots, \mc{P}_n}$. Let $\mc{Q}$ be a \Rq\  and $\mb{P'} = \seq{\semantics{\mc{P}_1}(\mc{Q}), \ldots, \semantics{\mc{P}_n}(\mc{Q})}$. The evaluation of  \PS\ is defined as follows:
\begin{equation}
\label{e:ext policyset semantics}
   \semantics{\mc{PS}}(\mc{Q}) = 
\begin{cases}
   \I_*      & \semantics{\mc{T}}(\mc{Q}) = \I \textrm{ and } \bigoplus_\comb(\mb{P'}) \in  \Set{\T_*,\ \I_*} \\
   \bot      & \semantics{\mc{T}}(\mc{Q}) = \F \textrm{ or}\\
             & \semantics{\mc{T}}(\mc{Q}) = \T \textrm{ and }\forall \mc{P}_i: \semantics{\mc{P}_i}(\mc{Q}) = \F \\
\bigoplus_\comb(\mb{P'}) & \textrm{otherwise} 
\end{cases}   
\end{equation}

\subsection{Example}
The following example simulate briefly how a policy is built using the abstraction. The  example is motivated by \cite{Evered2004,Hankin2009} which presents a health information system for a small nursing home in New South Wales, Australia. 
\begin{example}[Patient  Policy]
The general policy in the hospital in particular:
\begin{enumerate}
  \item Patient Record Policy
  \begin{itemize}
	 \item RP1: only designated patient \textbf{can} read his or her patient record except that if the patient is less than 18 years old, the patient's guardian is \textbf{permitted} also read the patient's record,
	 \item RP2: patients \textbf{may} only write patient surveys into their own records
	 \item RP3: both doctors and nurses are \textbf{permitted} to read any patient records,
  \end{itemize}
  \item Medical Record Policy
  \begin{itemize}
	 \item RM1: doctors \textbf{may} only write medical records for their own patients and 
	 \item RM2: \textbf{may not} write any other patient records,
  \end{itemize}
\end{enumerate}

The XACML policies for this example is shown in Figure \ref{f:example}. The   topmost policy in this example is the Patient Policy that contains two policies, namely the Patient Record Policy and the Medical Record Policy. The access is granted if either one of the Patient Record Policy or the Medical Record Policy gives a permit access. Thus in this case, we use permit-overrides combining algorithm to  combine those two policies. In order to restrict the access, each policy denies an access if there is a rule denies it. Thus, we use deny-overrides combining algorithms to combine the rules.

\begin{figure}[htb]
\vspace{-15pt}
\scriptsize{
\begin{Verbatim}[frame=single]
PS_patient = <Null, <P_patient_record, P_medical_record>, p-o>
P_patient_record = <Null, <RP1, RP2, RP3>, d-o>
P_medical_record = <Null, <RM1, RM2>, d-o> 

RP1 = 
< p,
  subject(patient) /\ action(read) /\ resource(patient_record),
  patient(id,X) /\ patient_record(id,Y) /\ 
  (X = Y \/ (age(Y) < 18 /\ guardian(X,Y))>

RP2 = 
< p,
  subject(patient) /\ action(write) /\ resource(patient_survey),
  patient(id,X) /\ patient_survey(id, X)>

RP3=
< p,
  (subject(doctor) \/ subject(nurse)) /\ action(read) /\ resource(patient_record), 
  true>

RM1 = 
< p,
  subject(doctor) /\ action(write) /\ resource(medical_record),
  doctor(id,X) /\ patient(id,Y) /\ medical_record(id, Y) /\ patient_doctor(Y,X)>

RM1 = 
< d,
  subject(doctor) /\ action(write) /\ resource(medical_record),
  doctor(id,X), patient(id,Y), medical_record(id, Y), not patient_doctor(Y,X)>
\end{Verbatim}
}
\vspace{-15pt} 
\caption{The XACML Policy for Patient Policy}
\label{f:example}
\end{figure}

Suppose now there is an emergency situation and a doctor $D$ asks  permission to read patient record $P$. The \Rq\ is as follows:
\begin{center}
\scriptsize{
\begin{BVerbatim}[frame=single] 
{ subject(doctor), action(read), resource(patient_record), 
  doctor(id,d), patient(id,p), patient_record(id,p)} 
\end{BVerbatim} 
}
\end{center}

Only \Tar\ \verb+RP3+ matches for this request and the effect of \verb+RP3+ is permit. Thus, the final result is doctor $D$ is allowed to read patient record $P$. Now, suppose that after doing some treatment, the doctor wants to update the medical record. A request is sent 
\begin{center}
\scriptsize{
\begin{BVerbatim}[frame=single] 
{ subject(doctor), action(write), resource(medical_record),  
  doctor(id,d), patient(id,p), medical_record(id,p)} 
\end{BVerbatim} 
}
\end{center}

 The \Tar\ \verb+RM1+ and the \Tar\ \verb+RM2+ match for this request, however because doctor $D$ is not registered as patient $P$'s doctor thus \Cnd\ \verb+RM1+ is evaluated to $\false$ while \Cnd\ \verb+RM2+ is evaluated to $\true$. In consequence, \R\ \verb+RM1+ is not applicable while \R\ \verb+RM2+ is applicable with effect deny.

\end{example}
      
% !TEX root = ../facs-extended.tex
%=================================Combining Algorithms=================================
\section{Combining Algorithms}
\label{s:combining algorithms}
Currently, there are four basic combining algorithms in XACML, namely (i) \textbf{permit-overrides}, (ii) \textbf{deny-overrides}, (iii) \textbf{first-applicable}, and (iv) \textbf{only-one-applicable}.  The input of a combining algorithm is a sequence of \R, \Pol\ or \PS\ values. In this section we give formalizations of the XACML 3.0 combining algorithms based on \cite{XACMLSpesification}. To guard against modelling artifacts we provide an alternative way of characterizing the policy combining algorithms and we formally prove the equivalence of these approaches.%\footnote{An extended version of this paper with all the proofs is available at \url{http://www2.imm.dtu.dk/~cdpu/Papers/the_logic_of_XACML-extended.pdf}.}

\subsection{Pairwise Policy Values}
\label{ss:pairwise policy values}
In $V_6$ we define the truth values of XACML components by extending $\T$ to $\Tp$ and $\Td$ and $\I$ to $\Id, \Ip$ and $\Idp$. This approach shows straightforwardly the status of XACML component. However, it is easier if we use numerical encoding when we need to do a computation, especially for computing policies compositions. Thus, we encode all the values returned by algorithms as pairs of natural numbers. 

In this numerical encoding, the value $\mathbf{1}$ represents an applicable value (either deny or permit), $\mathbf{\half}$ represents indeterminate value and $\mathbf{0}$ means there is no applicable value. In each tuple, the first element represents the \Deny\ value ($\Td$) and the later represents \Permit\ value ($\Tp$). We can say $[0,0]$ for not applicable ($\F$) because neither \Deny\ nor \Permit\ is applicable, $[1,0]$ for applicable with deny effect ($\Td$) because only \Deny\ value is applicable, $[\half,0]$ for $\Id$ because the \Deny\ part is indeterminate, $[\half,\half]$ for $\Idp$ because both \Deny\ and \Permit\  have indeterminate values. The conversion applies also for \Permit. 

A set of \textit{pairwise policy values} is $\policyvalues=\Set{[0,0], [\half,0], [0,\half],  [\half, \half], [1,0], [0,1]}$.  Let $[D,P]$ be an element in $\policyvalues$. We denote $\d([D,P]) = D$ and $\p([D,P]) = P$ for the function that returns the \Deny\ value and \Permit\ value, respectively. 

\label{d:k function}
We define $\delta : V_6 \ra \policyvalues$ as a mapping function that maps $V_6$ into $\policyvalues$ as follows:
\begin{equation}
\label{e:delta}
   \delta(X) = 
\begin{cases}
   [0,0]         & X = \F\\
    [\half,0]     & X = \Id\\
   [0,\half]     & X = \Ip\\
   [\half,\half] & X = \Idp\\
   [1,0]         & X = \T_{\deny} \\
   [0,1]         & X = \T_{\permit} \\
\end{cases}
\end{equation}
We define $\delta$ over a sequence $S$ as $\delta(S) = \seq{\delta(s) | s \in S }$.

We use pairwise comparison for the order of $\policyvalues$.   We define an order $\mysubseteq$ for $\policyvalues$ as follows $[D_1,P_1] \mysubseteq [D_2,P_2]$ iff $D_1 \leq D_2$ and $P_1 \leq P_2$ with $0\leq\half\leq 1$. We  write $\posetP$ for the partial ordered set (poset) $(\policyvalues, \mysubseteq)$ illustrated in Figure \ref{f:6-valued poset}.

\begin{figure}[h]
\begin{center}
\begin{tikzpicture}
\draw (0,0) -- (-3,2) -- (-3,4);
\draw (0,0) -- (3,2) -- (3,4);
\draw (-3,2) -- (0,4);
\draw (3,2) -- (0,4); 
\draw (0,0) node [fill=white] {$[0,0] = \F$};
\draw (-3,2) node [fill=white] {$[\half,0] = \Id$};
\draw (3,2) node [fill=white] {$[0,\half] = \Ip$};
\draw (-3,4) node [fill=white] {$[1,0] = \T_{\deny}$};
\draw (0,4) node [fill=white] {$[\half,\frac{1}{2}] = \Idp$};
\draw (3,4) node [fill=white] {$[0,1] = \T_{\permit}$};
\end{tikzpicture}
\end{center}
\vspace{-15pt} 
\caption{The Partial Ordered Set $\posetP$ for Pairwise Policy Values}
\label{f:6-valued poset}
\vspace{-15pt} 
\end{figure}

%\begin{definition}[Max and Min Functions]
Let $\max : 2^\mf{R} \ra \mf{R}$ be a function that returns the maximum value of a set of rational numbers and let $\min : 2^\mf{R} \ra \mf{R}$ be a function that returns the minimum value of a set of rational numbers. We define $\mymax : 2^\policyvalues \ra \policyvalues$ as a function that returns the maximum pairwise policy value which is defined as follows:
\begin{equation}
\mymax(S) = [\max(\Set{d(X) | X \in S}),\max(\Set{p(X) | X \in S})]
\end{equation}
and $\mymin : 2^\policyvalues \ra \policyvalues$ as a function that return the minimum pairwise policy value which is defined as follows: 
\begin{equation}
\mymin(S) = [\min(\Set{d(X) | X \in S}),\min(\Set{p(X) | X \in S})] 
\end{equation}
%\end{definition}

\subsection{Permit-Overrides Combining Algorithm}

The permit-overrides combining algorithm is intended for those cases where a permit decision should have priority over a deny decision. This algorithm (taken from \cite{XACMLSpesification}) has the following behaviour:
\begin{enumerate}
\item If any decision is $\Tp$ then the result is $\Tp$,
\item otherwise, if any decision is $\Idp$ then the result is $\Idp$,
\item otherwise, if any decision is $\Ip$  and another decision is $\Id$ or $\Td$, then the result is $\Idp$,
\item otherwise, if any decision is $\Ip$ then the result is $\Ip$,
\item otherwise, if decision is $\Td$ then the result is $\Td$,
\item otherwise, if any decision is $\Id$ then the result is $\Id$,
\item otherwise, the result is $\F$.
\end{enumerate}

\begin{figure}[t]
\begin{center}
\begin{tikzpicture}
\draw (0,0) -- (-1,1);
\draw (0,0) -- (1,1) -- (1,2);
\draw (-1,1) -- (0,3); 
\draw (1,2) -- (0,3) -- (0,4); 
\draw (0,0) node [fill=white] {$\F$};
\draw (-1,1) node [fill=white] {$\Ip$};
\draw (1,1) node [fill=white] {$\Id$};
\draw (1,2) node [fill=white] {$\Td$};
\draw (0,3) node [fill=white] {$\Idp$};
\draw (0,4) node [fill=white] {$\Tp$};

\draw (4,0) -- (5,1);
\draw (4,0) -- (3,1) -- (3,2);
\draw (5,1) -- (4,3); 
\draw (3,2) -- (4,3) -- (4,4); 
\draw (4,0) node [fill=white] {$\F$};
\draw (3,1) node [fill=white] {$\Ip$};
\draw (3,2) node [fill=white] {$\Tp$};
\draw (5,1) node [fill=white] {$\Id$};
\draw (4,3) node [fill=white] {$\Idp$};
\draw (4,4) node [fill=white] {$\T_{\deny}$};

\draw (8,0) -- (7,1) -- (7,2) -- (8,3);
\draw (8,0) -- (9,1) -- (9,2) -- (8,3);
\draw (8,0) node [fill=white] {$\F$};
\draw (7,1) node [fill=white] {$\T_{\deny}$};
\draw (7,2) node [fill=white] {$\Id$};
\draw (9,1) node [fill=white] {$\T_{\permit}$};
\draw (9,2) node [fill=white] {$\Ip$};
\draw (8,3) node [fill=white] {$\Idp$};

\end{tikzpicture}
\end{center}
\vspace{-15pt} 
\caption{The Lattice $\calL_{\po}$ for The Permit-Overrides Combining Algorithm (left), The Lattice $\calL_{\denyo}$ for The Deny-Overrides Combining Algorithm (middle) and The Lattice $\calL_{\oa}$ for The Only-One-Applicable Combining Algorithm (right)}
\label{f:l:permit-overrides}
\vspace{-15pt} 
\end{figure}
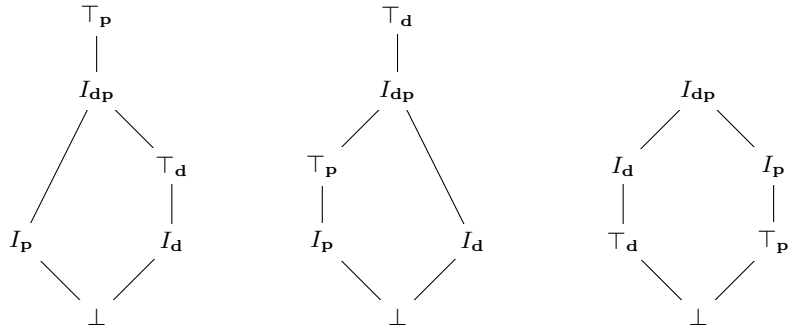

We call $\calL_{\po} = (V_6,\sqsubseteq_{\po})$ for the lattice using the permit-overrides combining algorithm  where  $\sqsubseteq_{\po}$ is the ordering depicted in  Figure \ref{f:l:permit-overrides}. The least upper bound operator for $\calL_{\po}$ is denoted by $\bigsqcup_{\po}$. 

\begin{definition}
The permit-overrides combining algorithm $\bigoplus_{\po}^{V_6}$ is a mapping function from a  sequence of $V_6$ elements into an element in $V_6$ as the result of composing policies. Let $S = \seq{s_1, \ldots,s_n}$ be a sequence of policy values in $V_6$ and $S' = \Set{s_1, \ldots, s_n}$. We define the \textit{permit-overrides combining algorithm under $V_6$} as follows:
   \begin{equation}
   \label{e:po}
   \bigoplus_{\po}^{V_6}(S) = \bigsqcup_{\po} S' \end{equation}
\end{definition}

The permit-overrides combining algorithm can also be expressed under $\policyvalues$. The idea is that we inspect the maximum value of \Deny\  and \Permit\  in the set of pairwise policy values. We conclude that the decision is permit if the \Permit\  is applicable (i.e. it has value 1). If the \Permit\  is indeterminate (i.e. it has value $\half$) then the decision is $\Idp$ if the \Deny\  is either indeterminate (i.e. it has value $\half$) or applicable (i.e. it has value 1). Otherwise we take the maximum value of \Deny\  and \Permit\  from the set of pairwise policy values as the result of permit-overrides combining algorithm.

\begin{definition}
The permit-overrides combining algorithm $\bigoplus_{\po}^{\policyvalues}$ is a mapping function from a  sequence of $\policyvalues$ elements into an element in $\policyvalues$ as the result of composing policies. Let $S = \seq{s_1, \ldots, s_n}$ be a sequence of pairwise policy values and $S' = \Set{s_1, \ldots, s_n}$. We define the  \textit{permit-overrides combining algorithm under $\policyvalues$} as follows:
\begin{equation}
\label{e:po pairwise}
   \bigoplus_{\po}^{\policyvalues}(S) = 
\begin{cases}
   [0,1]         & \mymax(S') = [\_,1] \\
   [\half,\half] & \mymax(S') = [D,\half], D \geq \half \\
   \mymax(S')     & \textrm{otherwise}

\end{cases}
\end{equation}
\end{definition}

\begin{proposition}
\label{prop:permitoverrides}
Let $S$ be a sequence of policy values in $V_6$. Then 
\[\delta(\bigoplus_{\po}^{V_6}(S)) = \bigoplus_{\po}^{\policyvalues}(\delta(S))\]
\end{proposition}

\begin{proof}
Let $S = \seq{s_1, \ldots, s_n}$ and $S' = \Set{s_1, \ldots, s_n}$. There are six possible outcomes for $\delta(\bigoplus_{\po}^{V_6}(S)) = \bigoplus_{\po}^{\policyvalues}(\delta(S))$:
\begin{enumerate}
   \item $\delta(\bigoplus_{\po}^{V_6}(S)) = [1,0]$ iff $\bigoplus_{\po}^{V_6}(S) = \Td = \bigsqcup_{\po} S'$ (by \eqref{e:po}). Based on $\sqsubseteq_{\po}$ we get that $\exists i : s_i = \Td$ and $\forall j : i \neq j, s_j \in \Set{\Td, \Id, \F}$. Thus, by \eqref{e:delta} we get that $\delta(s_i) = [1,0]$ and $\forall j : i \neq j, \delta(s_j) \in \Set{[1,0], [\half,0], [0,0]}$. Furthermore we get that $\mymax(\Set{\delta(s_1), \ldots, \delta(s_n)} = [1,0]$. Hence, by \eqref{e:po pairwise} we get that $\bigoplus_{\po}^{\policyvalues}(\delta(S)) = [1,0]$.

   \item $\delta(\bigoplus_{\po}^{V_6}(S)) = [0,1]$ iff $\bigoplus_{\po}^{V_6}(S) = \Tp = \bigsqcup_{\po} S'$ (by \eqref{e:po}). Based on $\sqsubseteq_{\po}$ we get that $\exists i : s_i = \Td$. Thus, by \eqref{e:delta} we get that $\delta(s_i) = [0,1]$. Furthermore we get $\mymax(\Set{\delta(s_1), \ldots, \delta(s_n)} = [\_,1])$. Hence, by \eqref{e:po pairwise} we get $\bigoplus_{\po}^{\policyvalues}(\delta(S)) = [0,1]$.

   \item $\delta(\bigoplus_{\po}^{V_6}(S)) = [\half,\half]$ iff $\bigoplus_{\po}^{V_6}(S) = \Idp = \bigsqcup_{\po} S'$ (by \eqref{e:po}). Based on $\sqsubseteq_{\po}$ there are three cases:
   \begin{enumerate}
      \item $\exists i : s_i = \Idp$ and $\forall j : j \neq i, s_j \in \Set{\Idp, \Ip, \Td, \Id, \F}$. Hence, by \eqref{e:delta} we get that $\delta(s_i) = [\half, \half]$ and $\forall s_j: \delta(s_j) \in \Set{[\half,\half], [0,\half], [1,0], [\half,0], [0,0]}$. Furthermore we get $\mymax(\Set{\delta(s_1), \ldots, \delta(s_n)} = [D,1])$ where $D \geq \half$. Therefore, by \eqref{e:po pairwise}  we get that  $\bigoplus_{\po}^{\policyvalues}(\delta(S)) = [\half,\half]$. 
      \item $\exists i, j : s_i = \Ip, s_j \in \Td$ and $\forall k: k \neq i, k \neq j, s_k \in \Set{\Ip, \Td, \Id, \F}$. Hence, by \eqref{e:delta} we get that $\delta(s_i) = [0,\half]$ and $\delta(s_j) = [1,0]$ and $\forall k: \delta(s_k) \in \Set{[0,\half], [1,0], [\half,0], [0,0]}$. Therefore, we get $\mymax(\Set{\delta(s_1), \ldots, \delta(s_n)} = [D,1])$ where $D \geq \half$. Moreover, by \eqref{e:po pairwise}  we get that  $\bigoplus_{\po}^{\policyvalues}(\delta(S)) = [\half,\half]$. 
      \item $\exists i, j : s_i = \Ip, s_j \in \Id$ and $\forall k : k \neq i, k \neq j, s_k \in \Set{\Ip, \Id, \F}$. Hence, by \eqref{e:delta} we get that $\delta(s_i) = [0,\half]$ and $\delta(s_j) = [1,0]$ and $\forall k: \delta(s_k) \in \Set{[0,\half], [\half,0], [0,0]}$. Hence, we get $\mymax(\Set{\delta(s_1), \ldots, \delta(s_n)} = [D,1])$ where $D \geq \half$. Moreover, by \eqref{e:po pairwise}  we get that  $\bigoplus_{\po}^{\policyvalues}(\delta(S)) = [\half,\half]$. 
   \end{enumerate}

   \item $\delta(\bigoplus_{\po}^{V_6}(S)) = [\half,0]$ iff $\bigoplus_{\po}^{V_6}(S) = \Id = \bigsqcup_{\po} S'$ (by \eqref{e:po}). Based on $\sqsubseteq_{\po}$ we get that $\exists i : s_i = \Id $ and $\forall j : j \neq i, s_j \in \Set{\Id, \F}$.  Hence, by \eqref{e:delta} we get that $\delta(s_i) = [\half,0]$ and $\forall j: \delta(s_j) \in \Set{[\half,0], [0,0]}$. Furthermore we get $\mymax(\Set{\delta(s_1), \ldots, \delta(s_n)} = [\half,0])$. Therefore, by \eqref{e:po pairwise}  we get that  $\bigoplus_{\po}^{\policyvalues}(\delta(S)) = [\half,0]$. 

   \item $\delta(\bigoplus_{\po}^{V_6}(S)) = [0, \half]$ iff $\bigoplus_{\po}^{V_6}(S) = \Ip = \bigsqcup_{\po} S'$ (by \eqref{e:po}). Based on $\sqsubseteq_{\po}$ we get that $\exists i : s_i = \Ip $ and $\forall j : j \neq i, s_j \in \Set{\Ip, \F}$.  Hence, by \eqref{e:delta} we get that $\delta(s_i) = [0,\half]$ and $\forall j: \delta(s_j) \in \Set{[\half,0], [0,0]}$. Furthermore we get $\mymax(\Set{\delta(s_1), \ldots, \delta(s_n)} = [0,\half])$. Therefore, by \eqref{e:po pairwise}  we get that  $\bigoplus_{\po}^{\policyvalues}(\delta(S)) = [0,\half]$.

   \item $\delta(\bigoplus_{\po}^{V_6}(S)) = [0,0]$ iff $\bigoplus_{\po}^{V_6}(S) = \F = \bigsqcup_{\po} S'$ (by \eqref{e:po}). Based on $\sqsubseteq_{\po}$ we get that $\forall i : s_i = \F$. Hence, by \eqref{e:delta} we get that $\forall i : \delta(s_i) = [0,0]$. Furthermore we get $\mymax(\Set{\delta(s_1), \ldots, \delta(s_n)} = [0,0])$. Therefore, by \eqref{e:po pairwise}  we get that  $\bigoplus_{\po}^{\policyvalues}(\delta(S)) = [0,0]$. \hfill $\Box$
\end{enumerate} 
\end{proof}

\subsection{Deny-Overrides Combining Algorithm}
The deny-overrides combining algorithm is intended for those cases where a deny decision should have priority over a permit decision. This algorithm (taken from \cite{XACMLSpesification}) has the following behaviour:
\begin{enumerate}
\item If any decision is $\Td$ then the result is $\Td$,
\item otherwise, if any decision is $\Idp$ then the result is $\Idp$,
\item otherwise, if any decision is $\Id$ and another decision is $\Ip$ or $\Tp$, then the result is $\Idp$,
\item otherwise, if any decision is $\Id$ then the result is $\Id$,
\item otherwise, if decision is $\Tp$ then the result is $\Tp$,
\item otherwise, if any decision is $\Ip$ then the result is $\Ip$,
\item otherwise, the result is $\F$.
\end{enumerate}

We call $\calL_{\denyo} = (V_6,\sqsubseteq_{\denyo})$ for the lattice using the deny-overrides combining algorithm where  $\sqsubseteq_{\denyo}$ is the ordering depicted in  Figure \ref{f:l:permit-overrides}. The least upper bound operator for $\calL_{\denyo}$ is denoted by $\bigsqcup_{\denyo}$. 

\begin{definition}
The deny-overrides combining algorithm $\bigoplus_{\denyo}^{V_6}$ is a mapping function from a  sequence of $V_6$ elements into an element in $V_6$ as the result of composing policies. Let $S = \seq{s_1, \ldots,s_n}$ be a sequence of policy values in $V_6$ and $S' = \Set{s_1, \ldots, s_n}$. We define the \textit{deny-overrides combining algorithm under $V_6$} as follows: 
\begin{equation}
   \label{e:do}
   \bigoplus_{\denyo}^{V_6}(S) = \bigsqcup_{\denyo} S' \end{equation}
\end{definition}

The deny-overrides combining algorithm can also be expressed under $\policyvalues$. The idea is similar to permit-overrides combining algorithm by symmetry. 
\begin{definition}
The deny-overrides combining algorithm $\bigoplus_{\denyo}^{\policyvalues}$ is a mapping function from a  sequence of $\policyvalues$ elements into an element in $\policyvalues$ as the result of composing policies. Let $S = \seq{s_1, \ldots,s_n}$ be a sequence of policy values in $\policyvalues$ and $S' = \Set{s_1, \ldots, s_n}$. We define the \textit{deny-overrides combining algorithm under $\policyvalues$} as follows: 
\begin{equation}
\label{e:do pairwise}
   \bigoplus_{\denyo}^{\policyvalues}(S) = 
\begin{cases}
   [1,0]         & \mymax(S') = [1,\_] \\
   [\half,\half] & \mymax(S') = [\half, P], P \geq \half \\
   \mymax(S')     & \textrm{otherwise}
\end{cases}
\end{equation}
\end{definition}

\begin{proposition}
\label{prop:denyoverrides}
Let $S$ be a sequence of policy values in $V_6$. Then 
\[\delta(\bigoplus_{\denyo}^{V_6}(S)) = \bigoplus_{\denyo}^{\policyvalues}(\delta(S))\]
\end{proposition}

The proof of Proposition \ref{prop:denyoverrides} is similar as the proof of Proposition \ref{prop:permitoverrides} by symmetry.

\subsection{First-Applicable Combining Algorithm}
The result of first-applicable algorithm is the first \R, \Pol\ or \PS\ element in the sequence whose \Tar\ and \Cnd\ is applicable. The pseudo-code of the first-applicable combining algorithm in XACML 3.0 \cite{XACMLSpesification} shows that the result of this algorithm is the first \R, \Pol\ or \PS\ that is not "not applicable". The idea is that there is a possibility an indeterminate policy could return to be an applicable policy. The first-applicable combining algorithm under $V_6$ and $\policyvalues$ are defined below.

\begin{definition}[First-Applicable Combining Algorithm]
The first-applicable combining algorithm $\bigoplus_{\fa}^{V_6}$ is a mapping function from a  sequence of $V_6$ elements into an element in $V_6$ as the result of composing policies. Let $S = \seq{s_1, \ldots,s_n}$ be a sequence of policy values in $V_6$. We define the \textit{first-applicable combining algorithm under $V_6$} as follows: 
\begin{equation}
\label{e:fa}
   \bigoplus_{\fa}^{V_6}(S) = 
\begin{cases}
   s_i & \exists i : s_i \neq \F \textrm{ and } \forall j < i: s_j = \F\\
   \F & \textrm{otherwise}
\end{cases}
\end{equation}
\end{definition}

\begin{definition}
The first-applicable combining algorithm $\bigoplus_{\fa}^{\policyvalues}$ is a mapping function from a  sequence of  $\policyvalues$ elements into an element in $\policyvalues$ as the result of composing policies. Let $S = \seq{s_1, \ldots,s_n}$ be a sequence of policy values in $\policyvalues$. We define the \textit{first applicable combining algorithm under $\policyvalues$} as follows:
\begin{equation}
\label{e:fa pairwise}
   \bigoplus_{\fa}^{\policyvalues}(S) = 
\begin{cases}
   s_i   & \exists i : s_i \neq [0,0] \textrm{ and } \forall j < i: s_j = [0,0] \\
   [0,0] & \textrm{otherwise}
\end{cases}
\end{equation}
\end{definition}

\begin{proposition}
\label{prop:firstapplicable}
Let $S$ be a sequence of policy values in $V_6$. Then 
\[\delta(\bigoplus_{\fa}^{V_6}(S)) = \bigoplus_{\fa}^{\policyvalues}(\delta(S))\]
\end{proposition}
\begin{proof}
The equation \eqref{e:fa} is the same as the equation \eqref{e:fa pairwise} when we consider the result of equation \eqref{e:fa} is mapped into $\policyvalues$ using $\delta$ function and the input of equation \eqref{e:fa pairwise} as $\delta(S)$. \hfill $\Box$
\end{proof}

\subsection{Only-One-Applicable Combining Algorithm}
The result of the only-one-applicable combining algorithm ensures that one and only one policy is applicable by virtue of their \Tar. If no policy applies, then the result is not applicable, but if more than one policy is applicable, then the result is indeterminate.  When exactly one policy is applicable, the result of the combining algorithm is the result of evaluating the single applicable policy.

We call $\calL_{\oa} = (V_6,\sqsubseteq_{\oa})$ for the lattice using the only-one-applicable combining algorithm where  $\sqsubseteq_{\oa}$ is the ordering depicted in  Figure \ref{f:l:permit-overrides}. The least upper bound operator for $\calL_{\oa}$ is denoted by $\bigsqcup_{\oa}$. 

\begin{definition}
The only-one-applicable combining algorithm $\bigoplus_{\oa}^{V_6}$ is a mapping function from a  sequence of $V_6$ elements into an element in $V_6$ as the result of composing policies. Let $S = \seq{s_1, \ldots,s_n}$ be a sequence of policy values in $V_6$ and $S' = \Set{s_1, \ldots, s_n}$. We define \textit{only-one-applicable combining algorithm under $V_6$} as follows 
\begin{equation}
\label{e:oa}
   \bigoplus_{\oa}^{V_6}(S) = 
\begin{cases}
   \Id & \exists i, j : i\neq j, s_i = s_j = \Td  \textrm{ and}\\
       & \forall k: s_k \neq \Td \ra s_k = \F\\
   \Ip & \exists i, j : i\neq j, s_i = s_j = \Tp  \textrm{ and}\\
       &  \forall k : s_k \neq \Tp \ra s_k = \F\\
   \bigsqcup_{\oa} S' & \textrm{otherwise}
\end{cases}
\end{equation}
\end{definition}

The only-one-applicable combining algorithm also can be expressed under $\policyvalues$. The idea is  that we inspect the maximum value of \Deny\ and \Permit\  returned from the given set of pairwise policy values. By inspecting the maximum value for each element, we know exactly the combination of pairwise policy values i.e.,  if we find that both \Deny\ and \Permit\  are not 0, it means that the \Deny\ and the \Permit\ are either applicable (i.e. it has value 1) or indeterminate (i.e. it has value $\half$). Thus, the result of this algorithm is $\Idp$ (based on the XACML 3.0 Specification \cite{XACMLSpesification}). However if only one element is not 0 then there is a possibility that many policies have the same applicable (or indeterminate) values. If there are at least two policies with the \Deny\ (or \Permit) are either applicable or indeterminate value, then the result is $\Id$ (or $\Ip$). Otherwise we take the maximum value of \Deny\  and \Permit\  from the given set of pairwise policy values as the result of only-one-applicable combining algorithm.

\begin{definition}
The only-one-applicable combining algorithm $\bigoplus_{\oa}^{\policyvalues}$ is a mapping function from a  sequence of $\policyvalues$ elements into an element in $\policyvalues$ as the result of composing policies. Let $S = \seq{s_1, \ldots,s_n}$ be a sequence of policy values in $\policyvalues$ and $S' = \Set{s_1, \ldots, s_n}$. We define \textit{only-one-applicable combining algorithm under $\policyvalues$} as follows 
\begin{equation}
\label{e:oa pairwise}
   \bigoplus_{\oa}^{\policyvalues}(S) = 
\begin{cases}
   [\half,\half] & \mymax(S') = [D,P], D,P \geq \half \\
   [\half,0]     & \mymax(S') = [D,0],  D \geq \half \textrm{ and } \\
                 & \exists i, j : i \neq j, \d(s_i), \d(s_j) \geq \half \\
   [0,\half]     & \mymax(S') = [0,P],  P \geq \half  \textrm{ and } \\
                 & \exists i, j : i \neq j, \p(s_i), \p(s_j) \geq \half\\
   \mymax(S')     & \textrm{otherwise} 
\end{cases}
\end{equation}
\end{definition}

\begin{proposition}
\label{prop:onlyoneapplicable}
Let $S$ be a sequence of policy values in $V_6$. Then 
\[\delta(\bigoplus_{\oa}^{V_6}(S)) = \bigoplus_{\oa}^{\policyvalues}(\delta(S))\]
\end{proposition}

\begin{proof}
Let $S = \seq{s_1, \ldots, s_n}$ and $S' = \Set{s_1, \ldots, s_n}$. There are six possible outcomes for $\delta(\bigoplus_{\oa}^{V_6}(S)) = \bigoplus_{\oa}^{\policyvalues}(\delta(S))$:
\begin{enumerate}
   \item $\delta(\bigoplus_{\oa}^{V_6}(S)) = [1,0]$ iff $\bigoplus_{\oa}^{V_6}(S) = \Td = \bigsqcup_{\oa} S'$ (by \eqref{e:oa}). Based on $\sqsubseteq_{\oa}$ we get that $\exists i : s_i = \Td$ and $\forall j : j \neq i, s_j = \F$. Furthermore, by \eqref{e:delta} we get that $\delta(s_i) = [1,0]$ and $\forall j : j \neq i, \delta(s_j) = [0,0]$. Therefore, $\mymax(\Set{\delta(s_1), \ldots, \delta(s_n)}) = [1,0]$. Thus, by \eqref{e:oa pairwise} we get $\bigoplus_{\oa}^{\policyvalues}(\delta(S)) = [1,0]$.

   \item $\delta(\bigoplus_{\oa}^{V_6}(S)) = [0,1]$ iff $\bigoplus_{\oa}^{V_6}(S) = \Tp = \bigsqcup_{\oa} S'$ (by \eqref{e:oa}). Based on $\sqsubseteq_{\oa}$ we get that $\exists i : s_i = \Tp$ and $\forall j : j \neq i, s_j = \F$. Furthermore, by\eqref{e:delta} we get that $\delta(s_i) = [0,1]$ and $\forall j: \delta(s_j) = [0,0]$. Hence, $\mymax(\Set{\delta(s_1), \ldots, \delta(s_n)}) = [0,1]$. Thus, by \eqref{e:oa pairwise} we get $\bigoplus_{\oa}^{\policyvalues}(\delta(S)) = [1,0]$.

   \item $\delta(\bigoplus_{\oa}^{V_6}(S)) = [\half,\half]$ iff $\bigoplus_{\oa}^{V_6}(S) = \Idp = \bigsqcup_{\oa} S'$ (by \eqref{e:oa}). Based on $\sqsubseteq_{\oa}$ there are two possibilities:
   \begin{enumerate}
      \item $\exists i : s_i = \Idp$. Hence, by \eqref{e:delta} we get that $\delta(s_i) = [\half,\half]$. Therefore, we get $\mymax(\Set{\delta(s_1), \ldots, \delta(s_n)}) = [D,P]$ where $D,P \geq \half$. Hence, by \eqref{e:oa pairwise} we get $\bigoplus_{\oa}^{\policyvalues}(\delta(S)) = [\half,\half]$.
      \item $\exists i : s_i \in \Set{\Id, \Td}$ and $\exists j : s_j \in \Set{\Ip, \Tp}$. Thus, by \eqref{e:delta} we get that $\delta(s_i) = [D,0]$ and $\delta(s_j) = [0,P]$ where $D,P \geq \half$. Furthermore, we get that $\mymax(\Set{\delta(s_1), \ldots, \delta(s_n)}) = [D,P]$ where $D,P \geq \half$. Hence, by \eqref{e:oa pairwise} we get $\bigoplus_{\oa}^{\policyvalues}(\delta(S)) = [\half,\half]$.
   \end{enumerate}

   \item $\delta(\bigoplus_{\oa}^{V_6}(S)) = [\half,0]$ iff $\bigoplus_{\oa}^{V_6}(S) = \Id$.  By \eqref{e:oa} we get that there are two possibilities:
   \begin{enumerate}
      \item $\exists i, j : i \neq j, s_i = s_j = \Td$ and $\forall k : s_k \neq \Td \ra s_k = \F$. Thus, by \eqref{e:delta} we get that $\delta(s_i) = \delta(s_j) = [1,0]$ and $\forall k : \delta(s_k) = [0,0]$. Hence,  $\mymax(\Set{\delta(s_1), \ldots, \delta(s_n)}) = [1,0]$ and we get $\p(s_i),\p(s_j) \geq \half$. Therefore, by \eqref{e:oa pairwise} we get $\bigoplus_{\oa}^{\policyvalues}(\delta(S)) = [\half,0]$.
      \item $\bigoplus_{\oa}^{V_6}(S) = \bigsqcup_{\oa} S' = \Id$. Thus, based on $\sqsubseteq_{\oa}$ we get that $\exists i : s_i = \Id$ and $\forall j : j \neq i, s_j \in \Set{\Id, \Td, \F}$. Thus, $\delta(s_i) = [\half,0]$ and $\forall j : \delta(s_j) \in \Set{[\half,0], [1,0], [0,0]}$  by \eqref{e:delta}. Hence, $\mymax(\Set{\delta(s_1), \ldots, \delta(s_n)}) = [D,0]$ where $D \geq \half$. There are two possibilities:
      \begin{enumerate}
         \item $D = 1$ iff $\exists k : s_k = [1,0]$. Thus, we get $s_i$ and $s_k$ where $\d(s_i),\d(s_k) \geq \half$. Therefore, by \eqref{e:oa pairwise} we get $\bigoplus_{\oa}^{\policyvalues}(\delta(S)) = [\half,0]$.
         \item $D = \half$. Therefore, by \eqref{e:oa pairwise} we get $\bigoplus_{\oa}^{\policyvalues}(\delta(S)) = [\half,0]$.
      \end{enumerate}
   \end{enumerate}

    \item $\delta(\bigoplus_{\oa}^{V_6}(S)) = [0,\half]$ iff $\bigoplus_{\oa}^{V_6}(S) = \Ip$.  By \eqref{e:oa} we get that there are two possibilities:
   \begin{enumerate}
      \item .$\exists i, j : i \neq j, s_i = s_j = \Tp$ and $\forall k : s_k \neq \Tp \ra s_k = \F$. Thus, by \eqref{e:delta} we get that $\delta(s_i) = \delta(s_j) = [0,1]$ and $\forall k: \delta(s_k) = [0,0]$. Hence,  $\mymax(\Set{\delta(s_1), \ldots, \delta(s_n)} = [0,1])$ and we get $\p(s_i),\p(s_j) \geq \half$. Therefore, by \eqref{e:oa pairwise} we get $\bigoplus_{\oa}^{\policyvalues}(\delta(S)) = [0,\half]$.
      \item $\bigoplus_{\oa}^{V_6}(S) = \bigsqcup_{\oa} S' = \Ip$. Thus, based on $\sqsubseteq_{\oa}$ we get that $\exists i : s_i = \Ip$ and $\forall j : i \neq j, s_j \in \Set{\Ip, \Tp, \F}$. Thus, $\delta(s_i) = [\half,0]$ and $\forall j: \delta(s_j) \in \Set{[0,\half], [0,1], [0,0]}$ by \eqref{e:delta}. Hence, $\mymax(\Set{\delta(s_1), \ldots, \delta(s_n)} = [0,P])$ where $P \geq \half$. There are two possibilities:
      \begin{enumerate}
         \item $P = 1$ iff $\exists k : s_k = [0,1]$. Thus, we get $s_i$ and $s_k$ where $\p(s_i),\p(s_k) \geq \half$. Therefore, by \eqref{e:oa pairwise} we get $\bigoplus_{\oa}^{\policyvalues}(\delta(S)) = [0,\half]$.
         \item $P = \half$. Therefore, by \eqref{e:oa pairwise} we get $\bigoplus_{\oa}^{\policyvalues}(\delta(S)) = [0,\half]$.
      \end{enumerate}
   \end{enumerate}

   \item $\delta(\bigoplus_{\oa}^{V_6}(S)) = [0,0]$ iff $\bigoplus_{\oa}^{V_6}(S) = \F = \bigsqcup_{\oa} S'$ (by \eqref{e:oa}). Based on $\sqsubseteq_{\oa}$ we get that $\forall i : s_i = \F$. Furthermore, by \eqref{e:delta} we get that $\forall i: \delta(s_i) = [0,0]$. Therefore, $\mymax(\Set{\delta(s_1), \ldots, \delta(s_n)} = [0,0])$. Thus, by \eqref{e:oa pairwise} we get $\bigoplus_{\oa}^{\policyvalues}(\delta(S)) = [0,0]$. \hfill $\Box$
\end{enumerate}
\end{proof}
 
%!TEX root = ../facs-extended.tex
%=================================Related Work=================================
\section{Related Work}
\label{s:related work}
We will focus the discussion on the formalization of XACML using Belnap logic \cite{Belnap1977} and \D-Algebra \cite{Ni2009} -- those two have a similar approach to the pairwise policy values approach explained in Section \ref{s:combining algorithms}.  In this section, we show the shortcoming of the formalization  on Bruns \etal work in \cite{Bruns2008} and Ni \etal work in \cite{Ni2009}.

\subsection{XACML Semantics under Belnap Four-Valued Logic}
Belnap in his paper \cite{Belnap1977} defines a four-valued logic over $\four = \Set{\topf,\tval, \fval, \botf}$. There are two orderings in Belnap logic, i.e., the knowledge ordering ($\leq_k$) and the truth ordering ($\leq_t$) (see Figure \ref{f:4-valued belnap logic}). 

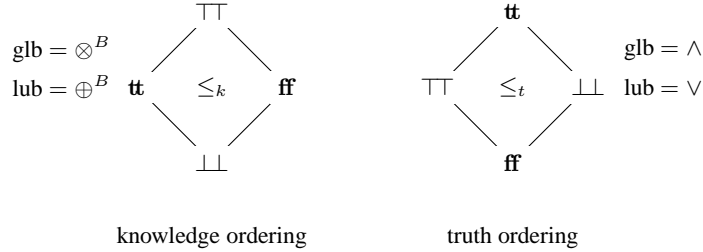
\begin{figure}[h]
\begin{center}
\begin{tikzpicture}
\draw (0,0) -- (-1,1) -- (0,2) -- (1,1) -- (0,0);
\draw (4,0) -- (3,1) -- (4,2) -- (5,1) -- (4,0);
\draw (0,0) node [fill=white] {$\botf$};
\draw (-1,1) node [fill=white] {$\tval$};
\draw (1,1) node [fill=white] {$\fval$};
\draw (0,2) node [fill=white] {$\topf$};
\draw (4,0) node [fill=white] {$\fval$};
\draw (3,1) node [fill=white] {$\topf$};
\draw (5,1) node [fill=white] {$\botf$};
\draw (4,2) node [fill=white] {$\tval$};
\draw (0,1) node [fill=white] {$\leq_{k}$};
\draw (4,1) node [fill=white] {$\leq_{t}$};
\draw (-2,1.5) node [fill=white] {glb $ = \otimesb$};
\draw (-2,1) node [fill=white] {lub $ = \oplusb$};
\draw (6,1.5) node [fill=white] {glb $ = \wedge$};
\draw (6,1) node [fill=white] {lub $ = \vee$};
\draw (0,-1)node [fill=white] {knowledge ordering};
\draw (4,-1)node [fill=white] {truth ordering};
\end{tikzpicture}
\end{center}
\vspace{-15pt} 
\caption{Bi-lattice of Belnap Four-Valued Logic}
\label{f:4-valued belnap logic}
\vspace{-15pt} 
\end{figure}

Bruns \etal  in PBel \cite{Bruns2007,Bruns2008} and also Hankin \etal  in AspectKB \cite{Hankin2009} use Belnap four-valued logic to represent the composition of access control policies. The responses of an access control system are $\tval$ when the policy is granted or access permitted, $\fval$ when the policy is not granted or access is denied, $\botf$ when there is no applicable policy and $\topf$ when conflict arises, i.e., an access is both permitted and denied. Additional operators are added as follows \cite{Bruns2008}:
\begin{itemize}
   \item overwriting operator $[y \mapsto z]$ with $y,z \in \four$. Expression $x[y \mapsto z]$ yields $x$ if $x \neq y$, and $z$ otherwise.
   \item priority operator $x > y$; it is a syntactic sugar of $x[\botf \mapsto y]$. 
\end{itemize}

Bruns \etal defined XACML combining algorithms using Belnap four-valued logic as follows \cite{Bruns2008}:
\begin{itemize}
   \item \textbf{permit-overrides}: $(p \oplusb q)[\topf \mapsto \fval]$
   \item \textbf{first-applicable}: $p > q$
   \item \textbf{only-one-applicable}: $(p \oplusb q)\oplusb((p \oplusb \neg p)\otimesb (q \oplusb \neg q))$
\end{itemize}

Bruns \etal suggested that the indeterminate value is treated as $\topf$. However, with 
indeterminate as $\topf$, the permit-overrides combining algorithm is not defined correctly. Suppose we have two policies: $p$ and $q$ where $p$ is permit and $q$ is indeterminate. The result of the permit-overrides combining algorithm is as follows
$(p \oplusb q)[\topf \mapsto \fval]  = (\tval \oplusb \topf)[\topf \mapsto \fval] = \topf[\topf \mapsto \fval] = \fval$. Based on the XACML 2.0 \cite{xacmlv2} and  the XACML 3.0 \cite{XACMLSpesification}, the result of permit-overrides combining algorithm should be permit ($\tval$). However, based on Belnap four-valued logic, the result is deny ($\fval$).

Bruns \etal tried to define indeterminate value as a conflict by formalizing it as $\topf$. However, their formulation of permit-overrides combining algorithm is inconsistent based on the standard XACML specification. Moreover, they said that sometimes indeterminate should be treated as $\botf$ and sometimes as $\topf$ \cite{Bruns2007}, but there is no explanation about under which circumstances that indeterminate is treated as $\topf$ or as $\botf$.  The treatment of indeterminate as $\topf$ is too strong because indeterminate does not always contains information about deny and permit in the same time. Only $\Idp$ contains information both deny and permit. However, $\Id$ and  $\Ip$  only contain information only about deny and permit, respectively. Even so, the value $\botf$ for indeterminate is too weak because indeterminate is treated as not applicable despite that there is information contained inside indeterminate value.
The Belnap four-valued logic has no explicit definition of indeterminate. In contrast, the Belnap four-valued has a \textit{conflict} value (i.e., $\topf$).

\subsection{XACML Semantics under \D-Algebra}
Ni \etal in \cite{Ni2009} define \D-algebra as a decision set together with some operations on it. 

\begin{definition}[\D-algebra \cite{Ni2009}]
Let $D$ be a nonempty set of elements, 0 be a constant element of $D$, $\neg$ be a unary operation on elements in \D, and $\oplusd,\otimesd$ be binary operations on elements in $D$. A \D-algebra is an algebraic structure $\langle D,\neg,\oplusd,\otimesd,0 \rangle$ closed on $\neg,\oplusd,\otimesd$ and satisfying the following axioms:
\begin{enumerate}
 \item $x \oplusd y = y \oplusd x$
 \item $(x \oplusd y)\oplusd z = x \oplusd(y \oplusd z)$
 \item $x \oplusd 0 = x$
 \item $\neg \neg x = x$
 \item $x \oplusd \neg 0 = \neg 0$
 \item $\neg(\neg x \oplusd y) \oplusd y = \neg(\neg y \oplusd x) \oplusd x$
 \item $x \otimesd y =  \begin{cases}
                          \neg 0 & : x = y \\
                          0 & : x \neq y
                         \end{cases}$
\end{enumerate}
\end{definition}

In order to write formulae in a compact form, for $x, y \in \mathcal{D}$, $x \odotd y = \neg(\neg x \oplusd \neg y)$ and $x \ominusd y = x \odotd \neg y$. 

Ni \etal  \cite{Ni2009} show that XACML decisions contain three different value, i.e., permit ($\{\pval\}$), deny ($\{\dval\}$) and not applicable ($\{\naa\}$). Those decision are \textit{deterministic decisions}. The \textit{non-deterministic decisions} such as $\Id$, $\Ip$ and $\Idp$ are denoted by $\Set{\dval, \naa}$, $\Set{\pval,\naa}$, and $\Set{\dval, \pval, \naa}$, respectively. The interpretation of  a \D-algebra on XACML decisions  is as follows \cite{Ni2009}:
\begin{itemize}
 \item $D$ is represented by $\calP(\Set{\pval, \dval, \naa})$
 \item 0 is represented by $\emptyset$
 \item $\neg x$ is represented by $\Set{\pval, \dval, \naa} - x$ where $x \in D$
 \item $x \oplusd y$ is represented by $x \cup y$ where $x,y \in D$
 \item $\otimesd$ is defined by axiom 7
\end{itemize}

%Sometimes  we call the \D-algebra on XACML as 8-valued \D-algebra because the cardinality of $D$ is 8. 
There are two values which are not in XACML, i.e., $\emptyset$ and $\Set{\pval, \dval}$. Simply we say $\emptyset$ for empty policy (or there is no policy) and $\Set{\pval, \dval}$ for a conflict. 

The composition function of permit-overrides using \D-Algebra is as follows:
\[
\begin{array}{lcl}
 f_{po}(x,y) & = & (x \oplusd y) \\
        &   & \ominusd (((x\otimesd\Set{\pval})\oplusd (y \otimesd \Set{\pval})) \odotd \Set{\dval, \naa}) \\
        &   & \ominusd (\neg((x \odotd y)\otimesd\Set{\naa})\odotd\Set{\naa}\odotd\neg((x\otimesd\emptyset)\oplusd(y\otimesd\emptyset)))
\end{array}
\]

%!TEX root = ../facs-extended.tex

The result of combining two policies using the permit-overrides combining algorithm over \D-Algebra can be seen in Table \ref{t:po D-algebra}.
\begin{table}[h]
\caption{Permit-Overrides Combining Algorithm Result Using \D-Algebra}
\label{t:po D-algebra}
\vspace{-15pt}
\begin{center}
\begin{tabular}{|c|c|c|c|c|c|c|c|c|}
\hline
$f_{po}(x,y)$ & $\emptyset$ & $\Set{\pval}$ & $\Set{\dval}$ & $\Set{\naa}$ & $\Set{\pval,\naa}$ & $\Set{\dval,\naa}$ & $\Set{\pval,\dval}$ & $\Set{\pval,\dval,\naa}$\\ \hline
$\emptyset$&$\emptyset$&$\Set{\pval}$&$\Set{\dval}$&$\Set{\naa}$&$\Set{\pval,\naa}$&$\Set{\dval,\naa}$&$\Set{\pval,\dval}$&$\Set{\pval,\dval,\naa}$\\ \hline
$\Set{\pval}$&$\Set{\pval}$&$\Set{\pval}$&$\Set{\pval}$&$\Set{\pval}$&$\Set{\pval}$&$\Set{\pval}$&$\Set{\pval}$&$\Set{\pval}$\\ \hline
$\Set{\dval}$&$\Set{\dval}$&$\Set{\pval}$&$\Set{\dval}$&$\Set{\dval}$&\mynote{$\Set{\pval,\dval}$}&$\Set{\dval}$&$\Set{\pval,\dval}$&\mynote{$\Set{\pval,\dval}$}\\ \hline
$\Set{\naa}$&$\Set{\naa}$&$\Set{\pval}$&$\Set{\dval}$&$\Set{\naa}$&$\Set{\pval,\naa}$&$\Set{\dval,\naa}$&$\Set{\pval,\dval}$&$\Set{\pval,\dval,\naa}$\\ \hline
$\Set{\pval,\naa}$&$\Set{\pval,\naa}$&$\Set{\pval}$&\mynote{$\Set{\pval,\dval}$}&$\Set{\pval,\naa}$&\mynote{$\Set{\pval}$}&$\Set{\pval,\dval,\naa}$&$\Set{\pval,\dval}$&\mynote{$\Set{\pval,\dval}$}\\ \hline
$\Set{\dval,\naa}$&$\Set{\dval,\naa}$&$\Set{\pval}$&$\Set{\dval}$&$\Set{\dval,\naa}$&$\Set{\pval,\dval,\naa}$&\mynote{$\Set{\dval}$}&$\Set{\pval,\dval}$&\mynote{$\Set{\pval,\dval}$}\\ \hline
$\Set{\pval,\dval}$&$\Set{\pval,\dval}$&$\Set{\pval}$&$\Set{\pval,\dval}$&$\Set{\pval,\dval}$&$\Set{\pval,\dval}$&$\Set{\pval,\dval}$&$\Set{\pval,\dval}$&$\Set{\pval,\dval}$\\ \hline
$\Set{\pval,\dval,\naa}$&$\Set{\pval,\dval,\naa}$&$\Set{\pval}$&\mynote{$\Set{\pval,\dval}$}&$\Set{\pval,\dval,\naa}$&\mynote{$\Set{\pval,\dval}$}&\mynote{$\Set{\pval,\dval}$}&$\Set{\pval,\dval}$&\mynote{$\Set{\pval,\dval}$}\\ \hline
\end{tabular}
\end{center}
\vspace{-15pt}
\end{table}

The composition function of deny-overrides using \D-Algebra is as follows:
\[
\begin{array}{lcl}
 f_{do}(x,y) & = & (x \oplusd y) \\
        &   & \ominusd (((x\otimesd\Set{\dval})\oplusd (y \otimesd \Set{\dval})) \odotd \Set{\pval, \naa}) \\
        &   & \ominusd (\neg((x \odotd y)\otimesd\Set{\naa})\odotd\Set{\naa}\odotd\neg((x\otimesd\emptyset)\oplusd(y\otimesd\emptyset)))
\end{array}
\]

The result of combining two policies using the deny-overrides combining algorithm over \D-Algebra can be seen in Table \ref{t:do D-algebra}.
\begin{table}[h]
\vspace{-15pt}
\caption{Deny-Overrides Combining Algorithm Result Using \D-Algebra}
\label{t:do D-algebra}
\vspace{-15pt}
\begin{center}
\begin{tabular}{|c|c|c|c|c|c|c|c|c|}
\hline
$f_{po}(x,y)$ & $\emptyset$ & $\Set{\pval}$ & $\Set{\dval}$ & $\Set{\naa}$ & $\Set{\pval,\naa}$ & $\Set{\dval,\naa}$ & $\Set{\pval,\dval}$ & $\Set{\pval,\dval,\naa}$\\ \hline

$\emptyset$&$\emptyset$&$\Set{\pval}$&$\Set{\dval}$&$\Set{\naa}$&$\Set{\pval,\naa}$&$\Set{\dval,\naa}$&$\Set{\pval,\dval}$&$\Set{\pval,\dval,\naa}$\\ \hline

$\Set{\pval}$&$\Set{\pval}$&$\Set{\pval}$&$\Set{\dval}$&$\Set{\pval}$&$\Set{\pval}$&\mynote{$\Set{\pval,\dval}$}&$\Set{\pval,\dval}$&\mynote{$\Set{\pval,\dval}$}\\ \hline

$\Set{\dval}$&$\Set{\dval}$&$\Set{\dval}$&$\Set{\dval}$&$\Set{\dval}$&$\Set{\dval}$&$\Set{\dval}$&$\Set{\dval}$&$\Set{\dval}$\\ \hline

$\Set{\naa}$&$\Set{\naa}$&$\Set{\pval}$&$\Set{\dval}$&$\Set{\naa}$&$\Set{\pval,\naa}$&$\Set{\dval,\naa}$&$\Set{\pval,\dval}$&$\Set{\pval,\dval,\naa}$\\ \hline

$\Set{\pval,\naa}$&$\Set{\pval,\naa}$&$\Set{\pval}$&$\Set{\dval}$&$\Set{\pval,\naa}$&\mynote{$\Set{\pval}$}&$\Set{\pval,\dval,\naa}$&$\Set{\pval,\dval}$&\mynote{$\Set{\pval,\dval}$}\\ \hline

$\Set{\dval,\naa}$&$\Set{\dval,\naa}$&\mynote{$\Set{\pval,\dval}$}&$\Set{\dval}$&$\Set{\dval,\naa}$&$\Set{\pval,\dval,\naa}$&\mynote{$\Set{\dval}$}&$\Set{\pval,\dval}$&\mynote{$\Set{\pval,\dval}$}\\ \hline

$\Set{\pval,\dval}$&$\Set{\pval,\dval}$&$\Set{\pval,\dval}$&$\Set{\dval}$&$\Set{\pval,\dval}$&$\Set{\pval,\dval}$&$\Set{\pval,\dval}$&$\Set{\pval,\dval}$&$\Set{\pval,\dval}$\\ \hline

$\Set{\pval,\dval,\naa}$&$\Set{\pval,\dval,\naa}$&\mynote{$\Set{\pval,\dval}$}&$\Set{\dval}$&$\Set{\pval,\dval,\naa}$&\mynote{$\Set{\pval,\dval}$}&\mynote{$\Set{\pval,\dval}$}&$\Set{\pval,\dval}$&\mynote{$\Set{\pval,\dval}$}\\ \hline
\end{tabular}
\end{center}
\vspace{-15pt}
\end{table}

The composition function of first-applicable using \D-Algebra is as follows:
\[
\begin{array}{lcl}
 f_{fa}(x,y) & = & (x \odotd (x \otimesd y)) \oplusd (y \odotd (x \otimesd \Set{\naa})) \oplusd (x \odotd (y \otimesd \Set{\naa})) \\
& & \oplusd(x \odotd (x \otimesd \Set{\pval})) \oplusd (x \odotd (x \otimesd \Set{\dval}))\\
& & \oplusd(x \odotd(x\otimesd\Set{\pval,\dval})) \oplusd (x \odotd(x\otimesd\Set{\pval,\dval,\naa}))\\
& & \oplusd(\Set{\pval} \odotd(x\otimesd\Set{\pval,\naa}) \odotd(y\otimesd\Set{\pval}))\\
& & \oplusd(\Set{\pval,\dval} \odotd(x\otimesd\Set{\pval,\naa}) \odotd (y \otimesd\Set{\dval})) \oplusd(y\otimesd\Set{\dval,\naa}))\\
& & \oplusd(\Set{\pval,\dval,\naa}\odotd(x \otimesd \Set{\pval,\naa}) \odotd(y\otimesd\Set{\dval,\naa}))\\
& & \oplusd(\Set{\dval}\odotd(x\otimesd\Set{\dval,\naa}) \odotd(y\otimesd\Set{\dval}))\\
& & \oplusd(\Set{\pval,\dval} \odotd(x \otimesd \Set{\dval,\naa}) \odotd((y\otimesd\Set{\pval}) \oplusd(y\otimesd\Set{\pval,\naa})))\\
& & \oplusd(\Set{\pval,\dval,\naa} \odotd (x \otimesd \Set{\dval,\naa}) \odotd (y \otimesd \Set{\pval,\naa}))
\end{array}
\]
The result of combining two policies using the first-applicable combining algorithm over \D-Algebra can be seen in Table \ref{t:fa D-algebra}.

\begin{table}[h]
\caption{First-Applicable Combining Algorithm Result Using \D-Algebra}
\label{t:fa D-algebra}
\vspace{-15pt}
\scriptsize{
\begin{center}
\begin{tabular}{|c|c|c|c|c|c|c|c|c|}
\hline
$f_{fa}(x,y)$ & $\emptyset$ & $\Set{\pval}$ & $\Set{\dval}$ & $\Set{\naa}$ & $\Set{\pval,\naa}$ & $\Set{\dval,\naa}$ & $\Set{\pval,\dval}$ & $\Set{\pval,\dval,\naa}$\\ \hline
$\emptyset$&$\emptyset$&\mynote{$\emptyset$}&\mynote{$\emptyset$}&\mynote{$\emptyset$}&\mynote{$\emptyset$}&\mynote{$\emptyset$}&\mynote{$\emptyset$}&\mynote{$\emptyset$}\\ \hline
$\Set{\pval}$&$\Set{\pval}$&$\Set{\pval}$&$\Set{\pval}$&$\Set{\pval}$&$\Set{\pval}$&$\Set{\pval}$&$\Set{\pval}$&$\Set{\pval}$\\ \hline
$\Set{\dval}$&$\Set{\dval}$&$\Set{\dval}$&$\Set{\dval}$&$\Set{\dval}$&$\Set{\dval}$&$\Set{\dval}$&$\Set{\dval}$&$\Set{\dval}$\\ \hline
$\Set{\naa}$&$\emptyset$&$\Set{\pval}$&$\Set{\dval}$&$\Set{\naa}$&$\Set{\pval,\naa}$&$\Set{\dval,\naa}$&$\Set{\pval,\dval}$&$\Set{\pval,\dval,\naa}$\\ \hline
$\Set{\pval,\naa}$&\mynote{$\emptyset$}&\mynote{$\Set{\pval}$}&\mynote{$\Set{\pval,\dval}$}&$\Set{\pval,\naa}$&$\Set{\pval,\naa}$&\mynote{$\Set{\pval,\dval,\naa}$}&\mynote{$\emptyset$}&\mynote{$\emptyset$}\\ \hline
$\Set{\dval,\naa}$&\mynote{$\emptyset$}&\mynote{$\Set{\pval,\dval}$}&\mynote{$\Set{\dval}$}&$\Set{\dval,\naa}$&\mynote{$\Set{\pval,\dval,\naa}$}&$\Set{\dval,\naa}$&\mynote{$\emptyset$}&\mynote{$\emptyset$}\\ \hline
$\Set{\pval,\dval}$&$\Set{\pval,\dval}$&$\Set{\pval,\dval}$&$\Set{\pval,\dval}$&$\Set{\pval,\dval}$&$\Set{\pval,\dval}$&$\Set{\pval,\dval}$&$\Set{\pval,\dval}$&$\Set{\pval,\dval}$\\ \hline
$\Set{\pval,\dval,\naa}$&$\Set{\pval,\dval,\naa}$&$\Set{\pval,\dval,\naa}$&$\Set{\pval,\dval,\naa}$&$\Set{\pval,\dval,\naa}$&$\Set{\pval,\dval,\naa}$&$\Set{\pval,\dval,\naa}$&$\Set{\pval,\dval,\naa}$&$\Set{\pval,\dval,\naa}$\\ \hline
\end{tabular}
\end{center}
}
\vspace{-15pt}
\end{table}

The composition function of only-one applicable using \D-Algebra is as follows:
\[
\begin{array}{lcl}
 f_{oo}(x,y) & = &  (x \odotd (y\otimesd\Set{\naa}))\oplusd(y \odotd (x\otimesd\Set{\naa}))
\end{array}
\]
The result of combining two policies using the only-one applicable combining algorithm over \D-Algebra can be seen in Table \ref{t:oo D-algebra}.
\begin{table}[h]
%\vspace{-15pt}
\caption{Only-One Applicable Combining Algorithm Result Using \D-Algebra}
\label{t:oo D-algebra}
\vspace{-15pt}
\begin{center}
\begin{tabular}{|c|c|c|c|c|c|c|c|c|}
\hline
$f_{oo}(x,y)$ & \hspace{2mm} $\emptyset$ \hspace{2mm} & $\Set{\pval}$ & $\Set{\dval}$ & $\Set{\naa}$ & $\Set{\pval,\naa}$ & $\Set{\dval,\naa}$ & $\Set{\pval,\dval}$ & $\Set{\pval,\dval,\naa}$\\ \hline
$\emptyset$&$\emptyset$&\mynote{$\emptyset$}&\mynote{$\emptyset$}&\mynote{$\emptyset$}&\mynote{$\emptyset$}&\mynote{$\emptyset$}&\mynote{$\emptyset$}&\mynote{$\emptyset$}\\ \hline
$\Set{\pval}$&\mynote{$\emptyset$}&\mynote{$\emptyset$}&\mynote{$\emptyset$}&$\Set{\pval}$&\mynote{$\emptyset$}&\mynote{$\emptyset$}&\mynote{$\emptyset$}&\mynote{$\emptyset$}\\ \hline
$\Set{\dval}$&\mynote{$\emptyset$}&\mynote{$\emptyset$}&\mynote{$\emptyset$}&$\Set{\dval}$&\mynote{$\emptyset$}&\mynote{$\emptyset$}&\mynote{$\emptyset$}&\mynote{$\emptyset$}\\ \hline
$\Set{\naa}$&$\emptyset$&$\Set{\pval}$&$\Set{\dval}$&$\Set{\naa}$&$\Set{\pval,\naa}$&$\Set{\dval,\naa}$&$\Set{\pval,\dval}$&$\Set{\pval,\dval,\naa}$\\ \hline
$\Set{\pval,\naa}$&\mynote{$\emptyset$}&\mynote{$\emptyset$}&\mynote{$\emptyset$}&$\Set{\pval,\naa}$&\mynote{$\emptyset$}&\mynote{$\emptyset$}&\mynote{$\emptyset$}&\mynote{$\emptyset$}\\ \hline
$\Set{\dval,\naa}$&\mynote{$\emptyset$}&\mynote{$\emptyset$}&\mynote{$\emptyset$}&$\Set{\dval,\naa}$&\mynote{$\emptyset$}&\mynote{$\emptyset$}&\mynote{$\emptyset$}&\mynote{$\emptyset$}\\ \hline
$\Set{\pval,\dval}$&\mynote{$\emptyset$}&\mynote{$\emptyset$}&\mynote{$\emptyset$}&$\Set{\pval,\dval}$&\mynote{$\emptyset$}&\mynote{$\emptyset$}&\mynote{$\emptyset$}&\mynote{$\emptyset$}\\ \hline
$\Set{\pval,\dval,\naa}$&\mynote{$\emptyset$}&\mynote{$\emptyset$}&\mynote{$\emptyset$}&$\Set{\pval,\dval,\naa}$&\mynote{$\emptyset$}&\mynote{$\emptyset$}&\mynote{$\emptyset$}&\mynote{$\emptyset$}\\ \hline
\end{tabular}
\end{center}
\vspace{-15pt}
\end{table}

As we can see in Table \ref{t:po D-algebra}, Table \ref{t:do D-algebra}, Table \ref{t:fa D-algebra} and Table \ref{t:oo D-algebra}, there are some results (indicated by red colour) that are different from the results based on the XACML specifications \cite{xacmlv2,XACMLSpesification}. In consequent, the combining algorithm functions under \D-algebra are not appropiate for XACML semantics. Their formulations are inconsistent based on the XACML 2.0 \cite{xacmlv2} and XACML 3.0 \cite{XACMLSpesification}.

%The composition function that Ni \etal proposed is inconsistent with neither the XACML 3.0 \cite{XACMLSpesification} nor the XACML 2.0 \cite{xacmlv2} as they claimed in \cite{Ni2009}.

Below we show an example that compares all of the results of permit-overrides combining algorithm under the logics discussed in this paper.

\begin{example}
\label{e:failure example}
Given two policies $P_1$ and $P_2$ where $P_1$ is Indeterminate Permit and $P_2$ is Deny. Let us use the permit-overrides combining algorithm to compose those two policies. Table \ref{t:result of example} shows the result of combining polices under Belnap logic, \D-algebra, $V_6$ and $\policyvalues$. 

\begin{table}[h]

\caption{Result of Permit-Overrides Combining Algorithm for Composing Two Policies $P_1$ and $P_2$ where $P_1$ is Indeterminate Permit and $P_2$ is Deny Under Various Logic}
\label{t:result of example}
\vspace{-10pt} 
\[
\begin{array}{|l|c|c|c|c|}
\hline
\textrm{Logic}       &P_1               & P_2    & \textrm{Permit-Overrides Function} & \textrm{Result} \\
\hline
\textrm{Belnap logic}& \topf            & \fval  &(\topf \oplusb \fval)[\topf \mapsto \fval]    &\fval \\
\textrm{\D-algebra}  & \Set{\pval, \naa}& \Set{\dval}  &f_{po}(\Set{\pval, \naa}, \Set{\dval})       &\Set{\pval, \dval}\\
V_6                  & \Ip              & \Td    &\bigoplus_{\po}^{V_6}(\seq{\Ip, \Td})        &\Idp\\
\policyvalues        & [0,\half]        & [1,0]  &\bigoplus_{\po}^{\policyvalues}(\seq{[0,\half], [1,0]})&[\half,\half]\\
\hline
\end{array}
\]
\end{table}
\end{example}

\vspace{-10pt}
The result of permit-overrides combining algorithm under Belnap logic is $\fval$ and under \D-algebra is $\Set{\pval, \dval}$. Under Bruns \etal approach using Belnap logic, the access is denied while under Ni \etal approach using \D-algebra, a conflict occurs. Both Bruns \etal  and Ni \etal claim that their approaches fit with XACML 2.0 \cite{xacmlv2}. Moreover \D-algebra claims that it fits with XACML 3.0 \cite{XACMLSpesification}. However based on  XACML 2.0 the result should be Indeterminate and based on XACML 3.0 the result should be Indeterminate Deny Permit and neither Belnap logic nor \D-algebra fits the specifications. We have illustrated that Belnap logic and \D-algebra in some cases give different result with the XACML specification. Conversely, our approaches  give consistent result based on the XACML 3.0 \cite{XACMLSpesification} and on the XACML 2.0 \cite{xacmlv2}.
%!TEX root = ../facs-extended.tex
%=================================Conclusion and Future Work=================================

\section{Conclusion}
\label{s:conclusion and future work}
%TODO: baca ulang bagian ini!!! :)

We have shown the formalization of XACML 3.0 step by step. We believe that with our approach, the user can understand better about how XACML works especially in the behaviour of combining algorithms. We show two approaches to formalizing standard XACML combining algorithms, i.e., using $V_6$ and $\policyvalues$. To guard against modelling artifacts, we formally prove the equivalence of these approaches. 

The pairwise policy values  approach is  useful in defining new combining algorithms. For example, suppose we have a new combining algorithm "all permit", i.e., the result of composing policies is permit if all policies give permit values, otherwise it is deny. Using pairwise policy values approach the result  of composing a set of policies values $S$  is permit ([0,1]) if  $\mymin(S) = [0,1] = \mymax(S)$, otherwise, it is deny ([1,0]).

Ni \etal\ proposes a \D-algebra over a set of decisions for XACML combining algorithms in \cite{Ni2009}. However, there are some mismatches between their results and the XACML specifications. Their formulations are inconsistent based both on the XACML 2.0 \cite{xacmlv2} and on the XACML 3.0 \cite{XACMLSpesification}. %\footnote{The detail of all of XACML decisions under \D-algebra can be seen in extended paper at \url{http://www2.imm.dtu.dk/~cdpu/Papers/the_logic_of_XACML-extended.pdf}.}

Both Belnap four-valued logic and \D-Algebra have a conflict value. In XACML, the conflict will never occur because the combining algorithms do not allow that. Conflict value might be a good indication that the policies are not well design. We propose an extended $\policyvalues$ which captures a conflict value in Appendix \ref{app:nine-valued logic}.

\bibliography{bibliography2}

\begin{thebibliography}{10}

\bibitem{xacml}
e{X}tensible {A}ccess {C}ontrol {M}arkup {L}anguage ({XACML}).
\newblock \url{http://xml.coverpages.org/xacml.html}.

\bibitem{xml}
{XML} 1.0 specification. w3.org. retrieved 2010-08-22.
\newblock \url{http://www.w3.org/TR/xml/}.

\bibitem{Ahn2010}
Gail-Joon Ahn, Hongxin Hu, Joohyung Lee, and Yunsong Meng.
\newblock Reasoning about xacml policy descriptions in answer set programming
  (preliminary report).
\newblock In {\em 13th International Workshop on Nonmonotonic Reasoning (NMR
  2010)}, 2010.

\bibitem{Belnap1977}
N.D. Belnap.
\newblock A useful four-valued logic.
\newblock In G.~Epstein and J.M. Dunn, editors, {\em Modern Uses of
  Multiple-Valued Logic}, pages 8--37. D. Reidel, Dordrecht, 1977.

\bibitem{Bruns2007}
Glenn Bruns, Daniel~S Dantas, and Michael Huth.
\newblock A simple and expressive semantic framework for policy composition in
  access control.
\newblock In {\em Proceedings of the 2007 ACM workshop on Formal methods in
  security engineering}, FMSE '07, pages 12--21, New York, NY, USA, 2007. ACM.

\bibitem{Bruns2008}
Glenn Bruns and Michael Huth.
\newblock Access-control via belnap logic: Effective and efficient composition
  and analysis.
\newblock In {\em 21st IEEE Computer Security Foundations Symposium}, June
  2008.

\bibitem{Evered2004}
Mark Evered and Serge B\"{o}geholz.
\newblock A case study in access control requirements for a health information
  systems.
\newblock In {\em Proceedings of the second workshop on Australasian
  information security, Data Mining and Web Intelligence, and Software
  Internationalisation - Volume 32}, ACSW Frontiers '04, pages 53--61,
  Darlinghurst, Australia, Australia, 2004. Australian Computer Society, Inc.

\bibitem{Halpern2008}
Joseph~Y. Halpern and Vicky Weissman.
\newblock Using first-order logic to reason about policies.
\newblock {\em ACM Transaction on Information and System Security (TISSEC)},
  11(4):1 -- 41, 2008.

\bibitem{Hankin2009}
Chris Hankin, Flemming Nielson, and Hanne~Riis Nielson.
\newblock Advice from belnap policies.
\newblock {\em Computer Security Foundations Symposium, IEEE}, 0:234--247,
  2009.

\bibitem{Kolovski2007}
Vladimir Kolovski and James Hendler.
\newblock Xacml policy analysis using description logics.
\newblock In {\em Proceedings of the 15th International World Wide Web
  Conference (WWW)}, 2007.

\bibitem{Kolovski2007a}
Vladimir Kolovski, James Hendler, and Bijan Parsia.
\newblock Formalizing xacml using defeasible description logics.
\newblock In {\em Proceedings of the 15th International World Wide Web
  Conference (WWW)}, 2007.

\bibitem{xacmlv2}
Tim Moses.
\newblock e{X}tensible {A}ccess {C}ontrol {M}arkup {L}anguage ({XACML}) version
  2.0.
\newblock Technical report, OASIS,
  http://docs.oasis-open.org/xacml/2.0/access\_control-xacml-2.0-core-spec-os.pdf,
  August 2010.

\bibitem{Ni2009}
Qun Ni, Elisa Bertino, and Jorge Lobo.
\newblock D-algebra for composing access control policy decisions.
\newblock In {\em ASIACCS '09: Proceedings of the 4th International Symposium
  on Information, Computer, and Communications Security}, pages 298--309, New
  York, NY, USA, 2009. ACM.

\bibitem{XACMLSpesification}
Erik Rissanen.
\newblock e{X}tensible {A}ccess {C}ontrol {M}arkup {L}anguage ({XACML}) version
  3.0 (committe specification 01).
\newblock Technical report, OASIS,
  http://docs.oasis-open.org/xacml/3.0/xacml-3.0-core-spec-cd-03-en.pdf, August
  2010.

\end{thebibliography}
%!TEX root = ../main.tex
\appendix
\section{Extended Pairwise Policy Values}
\label{app:nine-valued logic}
We add three values into $\policyvalues$, i.e. deny with indeterminate permit ($ [1,\half]$), permit with indeterminate deny ($[\half,1]$) and conflict ($[1,1]$) and we call 
the \textit{extended pairwise policy values} $\policyvalues_9 = \policyvalues \cup \Set{[1,\half], [\half,1], [1,1]}$. The extended pairwise policy values shows all possible combination of pairwise policy values.  The ordering of $\policyvalues_9$ is illustrated in Figure \ref{f:9-valued lattice}.

\begin{figure}[h]
\vspace{-15pt} 
\begin{center}
\begin{tikzpicture}[scale=0.8]
\draw (0,0) -- (-3,2) -- (-3,4) -- (-3,6) -- (0,8);
\draw (0,0) -- (3,2) -- (3,4) -- (3,6) -- (0,8);
\draw (-3,2) -- (0,4) -- (-3,6);
\draw (3,2) -- (0,4) -- (3,6); 
\draw (0,0) node [fill=white] {$[0,0] = \bot$};
\draw (-3,2) node [fill=white] {$[\frac{1}{2},0] = \ \Id$};
\draw (3,2) node [fill=white] {$[0,\frac{1}{2}] = \ \Ip$};
\draw (-3,4) node [fill=white] {$[1,0] = \ \T_\deny$};
\draw (0,4) node [fill=white] {$[\frac{1}{2},\frac{1}{2}] = \ \Idp$};
\draw (3,4) node [fill=white] {$[0,1] = \ \T_\permit$};
\draw (3,6) node [fill=white] {$[\half,1] = \Id\T_\permit$};
\draw (-3,6) node [fill=white] {$[1,\half] = \T_\deny\Ip$};
\draw (0,8) node [fill=white] {$[1,1] = \T_\deny\T_\permit$};
\end{tikzpicture}
\end{center}
\vspace{-15pt} 
\caption{Nine-Valued Lattice}
\label{f:9-valued lattice}
\vspace{-15pt} 
\end{figure}
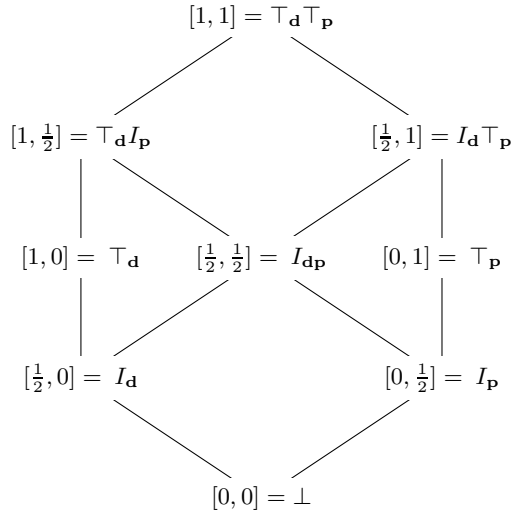

We can see that $\policyvalues_9$ forms a lattice (we call this $\calL_9$) where the top element is $[1,1]$ and the bottom element is $[0,0]$. The ordering of this lattice is the same as $\mysubseteq$ where the greatest lower bound and the least upper bound  for $S \subseteq \policyvalues_9$ are defined as follows:
\begin{equation*}
   \glbL_{\calL_9}S = \mymax(S) \textrm{ and } 
%\end{equation}
%\begin{equation}
   \lubL_{\calL_9}S = \mymin(S) 
\end{equation*}

\end{document}